%% file: main.tex
\documentclass{article}
\usepackage[preprint]{neurips_2026}

\usepackage[utf8]{inputenc}
\usepackage[T1]{fontenc}
\usepackage{hyperref}
\usepackage{url}
\usepackage{booktabs}
\usepackage{multirow}
\usepackage{longtable}
\usepackage{amsfonts}
\usepackage{amsmath}
\usepackage{amssymb}
\usepackage{nicefrac}
\usepackage{microtype}
\usepackage{graphicx}
\usepackage{enumitem}
\usepackage{xcolor}
\usepackage{algorithm}
\usepackage{algorithmic}
\usepackage{natbib}
\usepackage{amsthm}

\usepackage{tikz}
\usetikzlibrary{arrows.meta, positioning, calc}
\usepackage{xcolor}
\definecolor{inblue}{HTML}{DCE8F7}
\definecolor{inblueborder}{HTML}{5B84C4}

\definecolor{prepfill}{HTML}{F6E7CF}
\definecolor{prepborder}{HTML}{D89B1D}

\definecolor{modelfill}{HTML}{DDE9D7}
\definecolor{modelborder}{HTML}{7DA35A}

\definecolor{parserfill}{HTML}{F2D7D5}
\definecolor{parserborder}{HTML}{C46B61}

\definecolor{neutralfill}{HTML}{EFEFEF}
\definecolor{neutralborder}{HTML}{888888}

\definecolor{textdark}{HTML}{2B2B2B}
\definecolor{myblue}{HTML}{4C78A8}
\definecolor{mybluefill}{HTML}{EAF2FB}
\definecolor{myorange}{HTML}{F28E2B}
\definecolor{myorangefill}{HTML}{FDF1E3}
\definecolor{mygray}{HTML}{6B7280}
\definecolor{mygrayfill}{HTML}{F3F4F6}
\definecolor{mydark}{HTML}{1F2937}
\newtheorem{theorem}{Theorem}
\newtheorem{proposition}{Proposition}
\newtheorem{lemma}{Lemma}

\title{Calibrate, Don't Curate: \\ Label-Efficient Estimation from Noisy LLM Judges}

\author{
  Yanran Li \\
  Department of Biostatistics \\
  Columbia University \\
  \texttt{yl5465@columbia.edu} \\
}

\begin{document}

\maketitle

\begin{abstract}

Multi-judge evaluation is increasingly used to assess LLMs and reward models, and the prevailing heuristic is to curate: keep the most accurate judges and discard weaker ones. We show that this heuristic can reverse when the target is not point accuracy, but calibrated probabilistic evaluation from a labeled calibration set. Holding the aggregation and calibration procedures fixed, we compare accuracy-ranked top-$k$ judge selection with using the full judge panel. Across four labeled pairwise-evaluation benchmarks spanning LLM-as-judge and reward-model settings, the calibrated full panel consistently outperforms accuracy-based selection. On RewardBench2, retaining all judges achieves negative log-likelihood (NLL) of $0.006$ versus $0.013$ under top-5 selection, halving the calibration error. This advantage persists after judge-family deduplication and against stronger same-pipeline subset search. We explain this reversal with oracle analyses showing that the optimal calibrated risk under proper scoring rules cannot increase when additional judge signals are made available, and that even below-chance judges can be useful when their biases are learnable and their signals are non-redundant. The resulting operating principle is simple: in multi-judge evaluation with labeled calibration data, do not discard weak judges by accuracy alone; keep them when they are parseable, non-redundant, and calibratable.
\end{abstract}
\input{sections/introduction}
\input{sections/method}

\input{sections/theory}
\input{sections/experiments}
\input{sections/related_work}
\input{sections/conclusion}

\bibliographystyle{plainnat}
\bibliography{sections/references}

\appendix
\input{sections/appendix}


\end{document}

%% file: sections/introduction.tex
\section{Introduction}
\label{sec:intro}

Pairwise comparison by multiple judges is the dominant paradigm for evaluating large language models~\citep{chiang2024chatbot,zheng2023judging}, training reward models~\citep{ouyang2022training}, and aggregating crowdsourced annotations~\citep{chen2013crowd}.
These judges vary substantially in quality: LLM judges exhibit position and verbosity biases~\citep{zheng2023judging,ye2024justice}, and reward models range from near-random to near-expert~\citep{lambert2024rewardbench}.
The standard response is to \emph{curate}: select a panel of the best judges and discard the rest.
This is sensible when the goal is point accuracy~\citep{empirical2026rewardbench2,metajudges2025}. 
But accuracy is not the only goal.
Increasingly, evaluation systems must answer not just \emph{which response is better?}\ but \emph{how confident should we be?}\ Many evaluation settings require reliable confidence estimates, not just accurate winners. In this regime, pruning weak judges is not always the right choice: after calibration, keeping more judges can outperform selecting only the most accurate ones.

We study calibrated aggregation for labeled pairwise-comparison benchmarks, where held-out ground truth is available for calibration. These benchmarks, including JudgeBench~\citep{judgebench2025}, RewardBench~\citep{lambert2024rewardbench}, RewardBench2~\citep{malik2025rewardbench2}, and LLMBar~\citep{zeng2024llmbar}, raise a practical question: should we curate a small panel of high-accuracy judges, or calibrate the full panel? Across four benchmarks, calibrating the full panel usually outperforms top-$k$ accuracy-based pruning and remains competitive with stronger pruning baselines. The lesson is not to keep every judge blindly, but to avoid discarding judges solely because their point accuracy looks weak.

Why can weak judges help once calibration is available? The relevant question is not only whether a judge is accurate, but whether its votes carry learnable signal about the correct answer. A weak judge may still help estimate confidence: its disagreements can mark uncertain items, and a stable bias can be corrected or reversed. Section 3 studies this idea in an idealized setting where judge reliability is known; our experiments test whether the same principle remains useful when reliability must be estimated from limited calibration data.

This interpretation also explains when inclusion helps most. Accuracy-based pruning removes judges that look weak as point predictors, but their disagreements can still help identify uncertain items. Full-panel calibration keeps this information and learns how much to trust it. The gain is largest when the added judges contain learnable signal, and smallest when the panel is already redundant.

\subsection{Contributions}
\begin{enumerate}[leftmargin=*,itemsep=2pt]
    \item We present the first systematic evidence of a calibration reversal: full-panel calibration usually beats accuracy-based judge pruning in multi-judge LLM evaluation.

    \item We explain why weak judges can help: point accuracy is not the only criterion; learnable bias and disagreement can improve probability estimation.

    \item We stress-test inclusion against stronger curators, judge-family deduplication, and same-pipeline subset search; even the strongest curator does not improve on full-panel calibration.

    \item We identify post-hoc calibration as the main practical lever and characterize the boundary cases where inclusion can fail.
\end{enumerate}

%% file: sections/method.tex
\section{Setup and Methods}
\label{sec:method}

We consider $K$ judges evaluating $N$ labeled pairwise comparisons. Each judge $k \in [K]$ produces a verdict $y_{tk} \in \{0,1,-1\}$ for comparison $t$, where $-1$ denotes a missing verdict. A labeled calibration split provides ground-truth labels $z_t \in \{0,1\}$ for fitting the aggregation and calibration maps; a disjoint evaluation split is used only for reporting. Our goal is to estimate a calibrated predictive probability
$\hat{p}_t = \hat{P}(z_t = 1 \mid \{y_{tk}\}_{k=1}^K)$
such that $\mathbb{P}(z_t = 1 \mid \hat{p}_t = p) \approx p$, and optionally to construct valid prediction sets with coverage at least $1-\alpha$.

The big picture is a four-step statistical pipeline: (1) choose a judge panel and aggregate its votes into a raw probability; (2) correct residual bias; (3) apply a distribution-level calibration map; and (4) optionally wrap the result in conformal prediction sets. Figure~\ref{fig:overview} summarizes the comparison between accuracy-based pruning and full-panel inclusion. The inclusion-vs-selection experiments hold the aggregation and calibration families fixed across the two arms, so the only changed object is the judge panel. The nested model space
$\mathcal{M}_\omega \subset \mathcal{M}_{\omega,\delta} \subset \mathcal{M}_{\omega,\delta,G}$
keeps this progression explicit: each step expands the post-hoc model class without increasing calibrated log-loss asymptotically (Theorem~\ref{thm:calibration}).


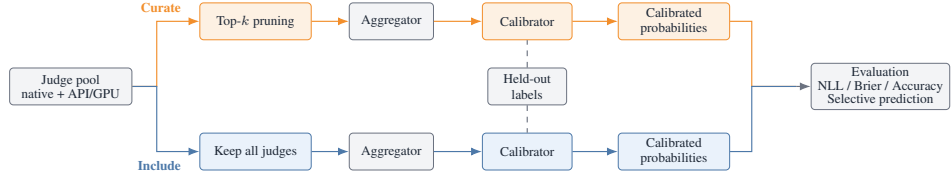
\begin{figure*}[t]
\centering
\resizebox{0.9\textwidth}{!}{
\begin{tikzpicture}[
    >=Latex,
    font=\scriptsize,
    flow/.style={-Latex, semithick},
    aux/.style={semithick, dashed, draw=mygray}, 
    merge/.style={semithick},
    base/.style={
        draw,
        rounded corners=2pt,
        align=center,
        minimum height=6.8mm,
        inner sep=2pt,
        text=mydark
    },
    shared/.style={base, draw=mygray, fill=mygrayfill},
    curate/.style={base, draw=myorange, fill=myorangefill},
    include/.style={base, draw=myblue, fill=mybluefill},
    lab/.style={font=\bfseries\scriptsize}
]

\coordinate (P0) at (0,0);          

\coordinate (P1) at (3.4,1.20);     
\coordinate (P2) at (3.4,-1.20);    

\coordinate (P3) at (5.9,1.20);     
\coordinate (P4) at (5.9,-1.20);    

\coordinate (P5) at (8.4,1.20);     
\coordinate (P6) at (8.4,-1.20);    

\coordinate (P7) at (11.1,1.20);    
\coordinate (P8) at (11.1,-1.20);   

\coordinate (PL) at (8.4,0.00);     

\coordinate (PJ) at (13.35,0.00);   
\coordinate (P9) at (14.9,0.00);    

\node[shared, minimum width=2.25cm] (pool) at (P0)
{Judge pool\\[-1pt] native + API/GPU};

\node[curate, minimum width=2.05cm] (prune) at (P1)
{Top-$k$ pruning};

\node[include, minimum width=2.05cm] (keep) at (P2)
{Keep all judges};

\node[shared, minimum width=1.55cm] (agg1) at (P3)
{Aggregator};

\node[shared, minimum width=1.55cm] (agg2) at (P4)
{Aggregator};

\node[curate, minimum width=1.65cm] (cal1) at (P5)
{Calibrator};

\node[include, minimum width=1.65cm] (cal2) at (P6)
{Calibrator};

\node[curate, minimum width=2.05cm] (out1) at (P7)
{Calibrated\\probabilities};

\node[include, minimum width=2.05cm] (out2) at (P8)
{Calibrated\\probabilities};

\node[shared, minimum width=1.45cm] (labels) at (PL)
{Held-out\\labels};

\node[shared, minimum width=2.55cm] (eval) at (P9)
{Evaluation\\[-1pt] NLL / Brier / Accuracy\\[-1pt] Selective prediction};

\node[lab, text=myorange, anchor=east] at (2.15,1.45) {Curate};
\node[lab, text=myblue,   anchor=east] at (2.15,-1.45) {Include};


\draw[flow, draw=myorange] (pool.east) -- ++(0.45,0) |- (prune.west);
\draw[flow, draw=myblue]   (pool.east) -- ++(0.45,0) |- (keep.west);

\draw[flow, draw=myorange] (prune.east) -- (agg1.west);
\draw[flow, draw=myorange] (agg1.east) -- (cal1.west);
\draw[flow, draw=myorange] (cal1.east) -- (out1.west);

\draw[flow, draw=myblue] (keep.east) -- (agg2.west);
\draw[flow, draw=myblue] (agg2.east) -- (cal2.west);
\draw[flow, draw=myblue] (cal2.east) -- (out2.west);

\draw[aux] (labels.north) -- (cal1.south);
\draw[aux] (labels.south) -- (cal2.north);

\draw[merge, draw=myorange] (out1.east) -- ++(0.40,0) |- (PJ);
\draw[merge, draw=myblue]   (out2.east) -- ++(0.40,0) |- (PJ);
\draw[flow, draw=mygray] (PJ) -- (eval.west);

\end{tikzpicture}
}
\caption{Overview of the curation and full-panel calibration pipelines. The two arms use the same aggregation and calibration families; only the judge-panel choice changes.}
\label{fig:overview}
\end{figure*}

\subsection{Step 1: Heterogeneity-Aware Judge Aggregation ($\mathcal{M}_\omega$)}
\label{sec:step1}

Step~1 maps a selected panel's verdict vector to a raw probability. Our default aggregator is an integrated Bayesian one-coin model, a closed-form implementation of the log-odds weighting principle in Proposition~\ref{prop:logodds}. From calibration data, judge $k$ receives $a_k \sim \mathrm{Beta}(c_k+1,n_k-c_k+1)$, where $c_k$ is its number of correct calibration verdicts and $n_k$ its number of non-missing calibration verdicts. Integrating over $a_k$ gives
\begin{equation}
\label{eq:integrated-posterior}
P(z_t = 1 \mid \mathbf{y}_t) = \sigma\!\left(\sum_{k \in J_t} \log \frac{\mathbb{E}[a_k^{y_{tk}}(1-a_k)^{1-y_{tk}}]}{\mathbb{E}[a_k^{1-y_{tk}}(1-a_k)^{y_{tk}}]}\right),
\end{equation}
where $J_t$ contains the non-missing judges for item $t$. Equivalently, judge $k$ contributes $\log(\alpha_k/\beta_k)$ for a positive verdict and $\log(\beta_k/\alpha_k)$ for a negative verdict, with $\alpha_k=c_k+1$ and $\beta_k=n_k-c_k+1$. Judges with uncertain accuracy have $\alpha_k\approx\beta_k$ and therefore near-zero weight. Dawid--Skene expectation-maximization (EM)~\citep{dawid1979maximum} fits the same role and behaves similarly after post-hoc calibration (Section~\ref{sec:cross-ablation}; Appendix~\ref{app:aggregator-comparison}). Curation happens before aggregation: top-$k$ ranks judges on the calibration split by posterior mean accuracy $\alpha_k/(\alpha_k+\beta_k)$, keeps the first $k$, and then uses the same aggregator and calibrator as the all-judges arm. Full inclusion sets $S=[K]$.

\subsection{Step 2: Residual Bias Correction ($\mathcal{M}_{\omega,\delta}$)}
\label{sec:step2}

Step~2 captures systematic bias left by the weighted aggregator with a low-capacity post-hoc map from the Step~1 score to the empirical label frequency:
\begin{equation}
\text{logit}(\hat{p}_t^{\text{corr}}) = a \cdot \text{logit}(\hat{p}_t) + b,
\end{equation}
where $\hat{p}_t$ is the Step~1 posterior and $(a,b)$ are fit by logistic regression on calibration labels. This Platt-style correction~\citep{platt1999probabilistic} captures global scale and shift bias. When item features $\mathbf{x}_t$ are available, such as topic or judge agreement level, we allow the residual map to vary with them:
\begin{equation}
\text{logit}(\hat{p}_t^{\text{corr}}) = a \cdot \text{logit}(\hat{p}_t) + \mathbf{x}_t^\top \boldsymbol{\gamma} + b.
\end{equation}
This residual correction enlarges the model class from
$\mathcal{M}_{\omega}$ to $\mathcal{M}_{\omega,\delta}$.
Since Step~1 is nested, the optimal asymptotic proper loss cannot increase;
Theorem~\ref{thm:calibration} formalizes this argument.

\subsection{Step 3: Distribution Calibration ($\mathcal{M}_{\omega,\delta,G}$)}
\label{sec:step3}

Step~3 applies a distribution-level calibration map to the corrected probability from Step~2. For binary outcomes we use beta calibration~\citep{kull2017beta},
\begin{equation}
\label{eq:beta-cal}
g(p) = \sigma\!\left(a \cdot \log p + b \cdot \log(1-p) + c\right),
\end{equation}
where $(a,b,c)$ are fit by minimizing negative log-likelihood (NLL) on the calibration split. This map generalizes the single-slope logistic correction in Step~2: Platt scaling is the tied-slope subfamily $b=-a$, the identity map is $(a,b,c)=(1,-1,0)$, and the untied slopes allow asymmetric corrections near confident positives and confident negatives. To keep this extra flexibility local to calibration, we regularize toward the identity with elastic net penalty $\lambda[\alpha(|a-1|+|b+1|+|c|)+(1-\alpha)((a-1)^2+(b+1)^2+c^2)]$, using $\alpha=0.5$ and fixing $\lambda=0.01$ in the main experiments. Lower-capacity calibrators and regularization sensitivity are evaluated in Appendix~\ref{app:matched-cal} and Appendix~\ref{app:additional}.

\subsection{Step 4: Conformal Coverage Guarantee}
\label{sec:step4}

Steps~1--3 output calibrated probabilities; Step~4 optionally wraps them with split conformal prediction~\citep{tibshirani2019conformal} to obtain finite-sample marginal coverage. The conformal calibration slice is disjoint from the data used to fit aggregation, residual correction, and beta calibration. Given the Step~3 probability $\hat{p}_t^{\text{final}}$, define the nonconformity score
\begin{equation}
s_t = 1 - \hat{p}_t^{\text{final}}(z_t) = \begin{cases} 1 - \hat{p}_t^{\text{final}} & \text{if } z_t = 1 \\ \hat{p}_t^{\text{final}} & \text{if } z_t = 0. \end{cases}
\end{equation}
Let $\hat{q}_{1-\alpha}$ be the $\lceil(1-\alpha)(n_{\mathrm{conf}}+1)\rceil/n_{\mathrm{conf}}$ empirical quantile of these scores on the conformal slice. The prediction set is
\begin{equation}
\mathcal{C}(x_t)=\{z:1-\hat{p}_t^{\text{final}}(z)\leq\hat{q}_{1-\alpha}\}.
\end{equation}

\subsection{Variance Decomposition}
\label{sec:decomposition}

Finally, we use a diagnostic variance decomposition to locate the main sources of evaluation uncertainty. This diagnostic is not used to fit Steps~1--4. We fit a two-parameter logistic (2PL) item response theory (IRT) model~\citep{lord1968irt},
$P(\text{correct}\mid j,i)=\sigma(\alpha_j(\theta_j-\beta_i))$, and summarize the induced logit-variance decomposition as
\begin{equation}
\underbrace{\text{Var}[\text{logit}(\hat{p}_t)]}_{\text{total}} = \underbrace{\text{Var}_j[\theta_j]}_{\text{judge noise}} + \underbrace{\text{Var}_i[\beta_i]}_{\text{item difficulty}} + 
\underbrace{\epsilon}_{\text{residual}}.
\end{equation}
This connects our multi-judge setting to the prediction-noise/data-noise decomposition of~\citet{wang2025noise}, and the dataset-level decomposition results are reported in Appendix~\ref{app:epi-alea}.

%% file: sections/theory.tex
\section{Theoretical Rationale for Inclusion}
\label{sec:theory}

Our theoretical contribution is a statistical account of \emph{why calibration
changes the judge selection problem}. We give two inclusion results: oracle-calibrated proper-loss risk is monotone in the judge set (Proposition~\ref{prop:monotone-risk}), and even anti-expert judges contribute positively to an exponentially decaying error bound (Theorem~\ref{thm:calibrated-cjt}). Both are oracle statements; Section~\ref{sec:experiments} closes the finite-sample gap empirically.

\begin{proposition}[Monotonicity of oracle-calibrated risk]
\label{prop:monotone-risk}
Let $Y \in \{0,1\}$, and let $X_k \in \{-1,+1\}$ be judge $k$'s label under the one-coin model with accuracy $p_k$ and conditional independence given $Y$.
For any subsets $S \subseteq T \subseteq [K]$, define the oracle-calibrated predictor $\eta_S(\mathbf{x}_S) := \Pr(Y{=}1 \mid \mathbf{X}_S = \mathbf{x}_S)$.
Then for both Brier score and NLL:
\begin{align}
R_{\mathrm{Brier}}(T) &= \mathbb{E}[\mathrm{Var}(Y \mid \mathbf{X}_T)] \leq \mathbb{E}[\mathrm{Var}(Y \mid \mathbf{X}_S)] = R_{\mathrm{Brier}}(S), \\
R_{\mathrm{NLL}}(T) &= H(Y \mid \mathbf{X}_T) \leq H(Y \mid \mathbf{X}_S) = R_{\mathrm{NLL}}(S).
\end{align}
Strict inequality holds if and only if $I(Y; \mathbf{X}_{T \setminus S} \mid \mathbf{X}_S) > 0$.
\end{proposition}

The proof is in Appendix~\ref{app:proof-monotone-risk}. Selection means choosing $S \subset T$, which discards conditional information and can only increase calibrated loss under any strictly proper scoring rule. Section~\ref{sec:selection-calibration} tests whether this monotonicity survives finite-sample estimation. The following theorem makes this quantitative: it gives an exponential error bound extending the Condorcet Jury Theorem to weak and anti-expert judges~(\citealp{condorcet1785essai}; see \citealp{berend1998sharp}).

\begin{theorem}[Calibrated Jury Theorem]
\label{thm:calibrated-cjt}
Under the one-coin model with labels encoded as $Y\in\{-1,+1\}$, equal priors $\Pr(Y{=}+1)=1/2$, and conditionally independent judge votes $X_k\in\{-1,+1\}$ satisfying $\Pr(X_k=Y\mid Y)=p_k\in(0,1)\setminus\{1/2\}$, the Bayes-optimal decision $\delta_K(\mathbf{x}) = \mathrm{sign}\{\sum_{k=1}^K \alpha_k x_k\}$, with arbitrary tie-breaking at zero and weights $\alpha_k = \log \frac{p_k}{1-p_k}$, satisfies:
\begin{equation}
\Pr(\delta_K \neq Y) \leq \tfrac{1}{2}\exp\!\Bigl(-\sum_{k=1}^K C_k\Bigr), \quad \text{where } C_k = -\tfrac{1}{2}\log\bigl(1-(2p_k - 1)^2\bigr) \geq \tfrac{1}{2}(2p_k-1)^2.
\end{equation}
In particular:
\begin{enumerate}[leftmargin=*,itemsep=1pt]
\item Judges with $p_k < 1/2$ contribute $C_k > 0$ (the optimal weight $\alpha_k < 0$ \emph{flips} their vote).
\item Only $p_k = 1/2$ exactly contributes $C_k = 0$ (pure noise).
\item If\; $\sum_{k=1}^\infty (2p_k-1)^2 = \infty$, then $\Pr(\delta_K \neq Y) \to 0$.
\end{enumerate}
\end{theorem}

The proof is in Appendix~\ref{app:proof-calibrated-cjt}. A judge with $p_k = 0.3$ contributes exactly as much as one with $p_k = 0.7$:
both have $|2p_k-1|=0.4$, and the optimal weight flips the anti-expert's vote
rather than discarding it. Section~\ref{sec:family-clustering} tests this prediction under judge correlation via family deduplication.

%% file: sections/experiments.tex
\section{Experiments}
\label{sec:experiments}

We evaluate calibration quality of multi-judge aggregation on four benchmarks where ground-truth labels exist. All sections share the same datasets, judge panels, metrics, and split protocol. Section~\ref{sec:main-calibration} tests the main inclusion claim, Section~\ref{sec:stronger-baselines} rules out stronger pruning baselines, and Section~\ref{sec:cross-ablation} isolates aggregator and calibrator effects. Section~\ref{sec:scope-and-budget} then studies label budget and scope.
\subsection{Setup}
\label{sec:setup}

\textbf{Datasets.}
We evaluated four pairwise-comparison benchmarks with ground-truth labels: \emph{JudgeBench}~\citep{judgebench2025} ($350$ pairs with objectively correct answers, $32$ native LLM judges),
\emph{RewardBench}~\citep{lambert2024rewardbench} ($2985$ chosen-vs-rejected pairs, $100$ reward models),
\emph{RewardBench~2}~\citep{malik2025rewardbench2} ($1865$ harder pairs, $174$ reward models), and
\emph{LLMBar}~\citep{zeng2024llmbar} ($419$ instruction-following pairs including adversarial subsets, $44$ native pairwise evaluators).

\textbf{Judge panels.}
We augment each benchmark with API- and/or GPU-hosted LLM judges to widen backbone-family coverage, yielding $K{=}38$ (JudgeBench), $104$ (RewardBench), $174$ (RewardBench~2), and $53$ (LLMBar).
RewardBench~2 retains its native $174$ reward models without additional judges; GPU-judge stress tests for JudgeBench and RewardBench~2 are reported as scope diagnostics.
All added judges use a single Multi-Turn Bench (MT-bench)-style pairwise prompt~\citep{zheng2023judging}; A/B order is randomized per item, verdicts are extracted by a four-level fallback parser, and parse failures are recorded as missing and skipped by the aggregator.
Per-judge identities, panel composition, prompts, parser, decoding settings, and verdict coverage are in Appendix~\ref{app:judge-list} (Table~\ref{tab:judge-panel}).

\textbf{Methods.}
We compare uncalibrated aggregators (majority vote, accuracy-weighted, log-odds weighted, Dawid--Skene EM, Bayesian one-coin, and top-$k$ + one-coin), their post-hoc calibrated variants (Platt scaling, temperature scaling, beta calibration, and isotonic regression), and split-conformal wrappers at 90\% and 80\% target coverage.\footnote{Each calibrator is applied to every aggregation method; when summarizing a calibrator family, we report the best aggregation--calibration combination for that calibrator.}

\textbf{Metrics.}
Our headline metrics are NLL and Brier score, both strictly proper scoring rules~\citep{gneiting2007strictly} that directly measure calibration quality without binning:
\begin{align*}
\mathrm{NLL} &= -\frac{1}{N}\sum_t \left[z_t \log \hat{p}_t + (1-z_t)\log(1-\hat{p}_t)\right], &
\mathrm{Brier} &= \frac{1}{N}\sum_t (\hat{p}_t-z_t)^2.
\end{align*}
We also report expected calibration error (ECE), point accuracy, conformal coverage, and prediction-set size:
\begin{align*}
\mathrm{ECE} &= \sum_{b=1}^{10}\frac{|I_b|}{N}\left|\mathrm{acc}(I_b)-\mathrm{conf}(I_b)\right|, &
\mathrm{Acc} &= \frac{1}{N}\sum_t \mathbf{1}\{\mathbf{1}(\hat{p}_t\ge0.5)=z_t\}, \\
\mathrm{Coverage} &= \frac{1}{N}\sum_t \mathbf{1}\{z_t\in\mathcal{C}(x_t)\}, &
\mathrm{Size} &= \frac{1}{N}\sum_t |\mathcal{C}(x_t)|.
\end{align*}
For ECE, we use 10 equal-width confidence bins with $\mathrm{conf}_t=\max\{\hat{p}_t,1-\hat{p}_t\}$ after thresholding at $0.5$; rankings are unchanged under the scaling-binning estimator of~\citet{kumar2019verified}.
Lower NLL, Brier, ECE, and set size are better; higher accuracy and coverage are better.
Unless otherwise stated, every experiment uses a 50/50 calibration/evaluation split repeated over 100 independent random permutations, fits calibration methods only on the calibration split, and clips probabilities to $[0.001,0.999]$ for numerical stability.
Statistical tests are reported with each comparison.

\subsection{Main Result: Calibrated Inclusion Beats Selection}
\label{sec:main-calibration}
\label{sec:selection-calibration}

We first separate raw aggregation quality from calibrated probability quality.
Without post-hoc calibration, pruning weak judges can improve raw accuracy and sometimes raw NLL by removing extreme low-quality signals (Table~\ref{tab:uncalibrated}; Appendix~\ref{app:uncalibrated-selection}). After applying the same beta calibrator, the ranking reverses: using all judges achieves the lowest NLL in the primary Top-3 and Top-5 comparisons on all four datasets, with one Top-10 boundary case on LLMBar (Figure~\ref{fig:inclusion-vs-selection}a).
Concretely, Top-5 selection increases calibrated NLL relative to all judges by $+19\%$ on JudgeBench ($0.019$ vs.\ $0.016$), $+31\%$ on RewardBench ($0.007$ vs.\ $0.006$), $+123\%$ on RewardBench~2 ($0.013$ vs.\ $0.006$), and $+22\%$ on LLMBar ($0.014$ vs.\ $0.012$).

\begin{table}[t]
\centering
\caption{Uncalibrated results on JudgeBench (left; 38 judges, 350 pairs) and RewardBench (right; 104 judges, 2985 pairs). All rows use Bayesian one-coin \emph{without} post-hoc beta calibration.}
\label{tab:uncalibrated}
\footnotesize
\setlength{\tabcolsep}{4pt}
\begin{minipage}[t]{0.50\textwidth}
\centering
\begin{tabular}{@{}lcccc@{}}
\toprule
Method & NLL & Brier & ECE & Acc\,(\%) \\
\midrule
Majority Vote & 0.568 & 0.193 & 0.372 & 70.4 \\
Accuracy-Weighted & 0.557 & 0.188 & 0.365 & 70.7 \\
Temperature Scaling & 0.570 & 0.193 & 0.385 & 70.4 \\
Dawid-Skene EM & 1.682 & 0.278 & 0.294 & 70.9 \\
Log-Odds Weighted & 1.231 & 0.231 & 0.261 & 73.9 \\
Bayesian One-Coin (all) & 1.224 & 0.229 & 0.261 & 73.9 \\
\midrule
Top-3 + Bayesian & \textbf{0.501} & \textbf{0.148} & 0.233 & \textbf{80.3} \\
Top-5 + Bayesian & 0.589 & 0.159 & \textbf{0.224} & 79.3 \\
Top-10 + Bayesian & 0.801 & 0.182 & 0.226 & 78.5 \\
\bottomrule
\end{tabular}
\end{minipage}%
\hfill
\begin{minipage}[t]{0.50\textwidth}
\centering
\begin{tabular}{@{}lcccc@{}}
\toprule
Method & NLL & Brier & ECE & Acc\,(\%) \\
\midrule
Majority Vote & 0.348 & 0.099 & 0.269 & 90.2 \\
Accuracy-Weighted & 0.314 & 0.086 & 0.244 & 91.0 \\
Temperature Scaling & 0.263 & 0.077 & 0.152 & 90.2 \\
Dawid-Skene EM & 0.470 & 0.071 & 0.074 & 92.7 \\
Log-Odds Weighted & 0.415 & 0.063 & 0.065 & 93.6 \\
Bayesian One-Coin (all) & 0.411 & 0.062 & 0.065 & 93.6 \\
\midrule
Top-3 + Bayesian & \textbf{0.200} & 0.040 & 0.044 & 95.7 \\
Top-5 + Bayesian & 0.214 & \textbf{0.039} & \textbf{0.042} & \textbf{95.9} \\
Top-10 + Bayesian & 0.257 & 0.041 & 0.044 & 95.8 \\
\bottomrule
\end{tabular}
\end{minipage}
\end{table}

\begin{figure}[t]
\centering
\includegraphics[width=\textwidth]{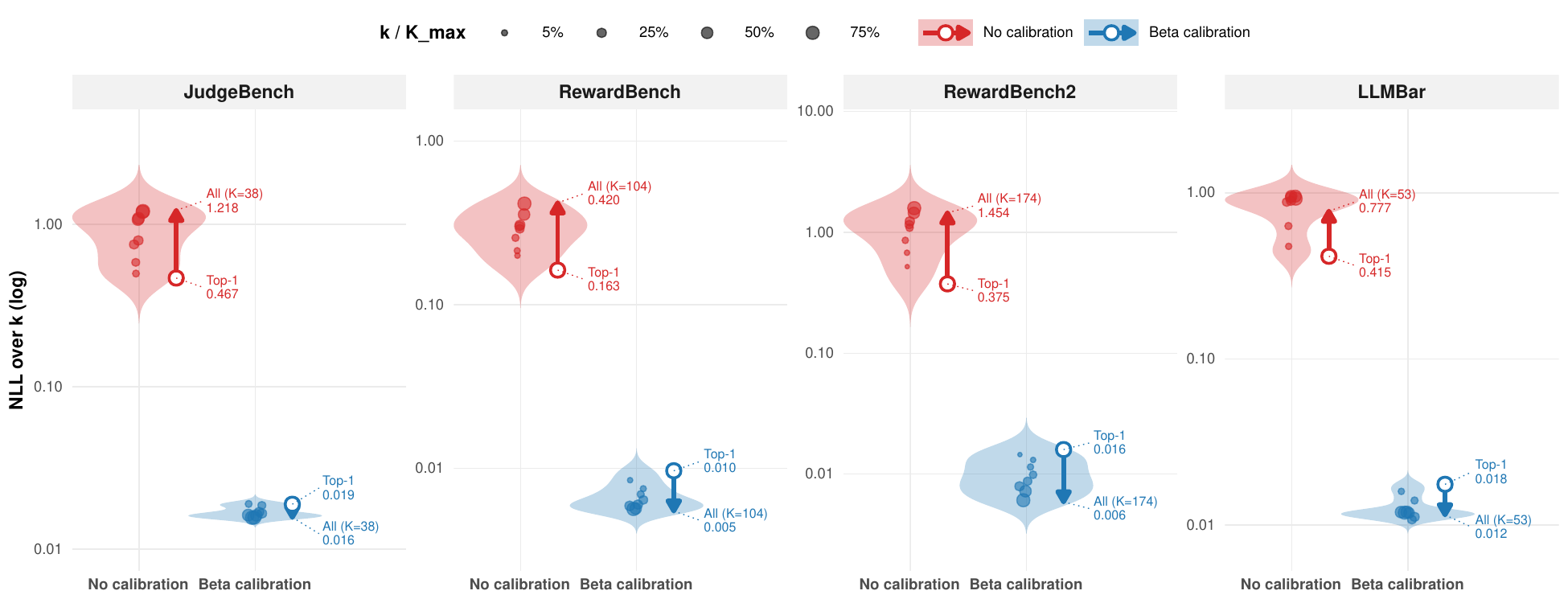}
\caption{Calibration flips the selection--inclusion tradeoff. 
Across datasets, top-\(k\) pruning can reduce uncalibrated NLL, but after beta calibration the all-judge panel achieves the lowest NLL. }
\label{fig:inclusion-vs-selection}
\end{figure}

Selection improves point accuracy by removing weak judges, but calibration benefits from the disagreement they contribute: disagreement marks difficult items and produces less extreme probabilities. In practice, inclusion should be the default unless the scope diagnostics in Section~\ref{par:scope-failures} are triggered. The factor of variability is the random calibration/evaluation split. Our primary effect size is the mean per-item NLL difference $\Delta = \mathrm{NLL}_{\mathrm{top}k} - \mathrm{NLL}_{\mathrm{all}}$ with bootstrap $95\%$ confidence intervals (CIs) over splits, and our primary test is the item-level paired permutation test with 10{,}000 random sign flips within each split. We report the median test statistic and median $p$-value across splits, with inference details in Appendix~\ref{app:headline-inference}; the simultaneous-across-budgets claim (Sec.~\ref{sec:scope-and-budget}) uses a separate Gaussian multiplier bootstrap (Appendix~\ref{app:joint-test}).
Table~\ref{tab:topk-inference} summarizes results across datasets and top-$k$ values.

\begin{table}[t]
\centering
\caption{Paired-permutation evidence for calibrated inclusion over top-$k$ selection.}
\label{tab:topk-inference}
\small
\begin{tabular}{@{}l ccc@{}}
\toprule
& Top-3 & Top-5 & Top-10 \\
\midrule
JudgeBench & $\Delta{=}0.003$, $t{=}11.7^{***}$ & $\Delta{=}0.003$, $t{=}10.8^{***}$ & $\Delta{=}0.001$, $t{=}4.0^{***}$ \\
RewardBench & $\Delta{=}0.003$, $t{=}27.2^{***}$ & $\Delta{=}0.002$, $t{=}18.3^{***}$ & $\Delta{=}0.001$, $t{=}12.9^{***}$ \\
RewardBench~2 & $\Delta{=}0.009$, $t{=}71.6^{***}$ & $\Delta{=}0.007$, $t{=}57.8^{***}$ & $\Delta{=}0.006$, $t{=}44.4^{***}$ \\
LLMBar & $\Delta{=}0.004$, $t{=}16.6^{***}$ & $\Delta{=}0.003$, $t{=}9.8^{***}$ & $\Delta{=}{-}0.001$, $t{=}{-}2.7^{\text{ns}}$ \\
\bottomrule
\end{tabular}\\[2pt]
\parbox{0.95\linewidth}{\footnotesize $\Delta$: mean per-item NLL(Top-$k$) $-$ NLL(All) (median across 100 splits); $t$: median item-level paired permutation $t$-statistic. $^{***}$: $p < 0.001$; $^{\text{ns}}$: not significant under 10K sign flips.}
\end{table}

Bootstrap 95\% CIs for $\Delta$ exclude zero on every comparison except LLMBar Top-10, supporting Proposition~\ref{prop:monotone-risk}: oracle monotonicity holds in finite samples.\label{sec:family-clustering} Proposition~\ref{prop:monotone-risk} and Theorem~\ref{thm:calibrated-cjt} assume conditional independence, but in practice same-family judges are correlated.
We cluster JudgeBench's 38 judges into 9 families (OpenAI, Anthropic, Meta, Google, Alibaba, Mistral, JudgeLM, Prometheus, RewardModel) and select one representative per family by calibration-set accuracy.
All judges ($K{=}38$, NLL $= 0.0157$) outperform one-per-family ($K{=}9$, NLL $= 0.0185$; $\Delta = 0.003$, $t = 10.0$, $p < 0.001$).
Inclusion beats selection even after removing within-family redundancy, supporting Theorem~\ref{thm:calibrated-cjt}'s prediction that cross-family judges carry conditional information.

\subsection{Stronger Curators Tie Inclusion at Best}
\label{sec:stronger-baselines}

Top-$k$ by calibration-set accuracy is a simple pruning rule, so we test whether stronger pruning-style curators can beat all-judges calibration.
We compare against greedy forward selection on inner-split calibrated NLL and $\ell_1$-sparse logistic regression on the raw vote matrix; oracle top-$5$ and random-$5$ serve as sanity checks for accuracy-ranked pruning, not as calibration-aware baselines. Full curator definitions, regularization choices, and subset-search protocol details are in Appendix~\ref{app:curator-details}.

\begin{figure}[t]
\centering
\includegraphics[width=\textwidth]{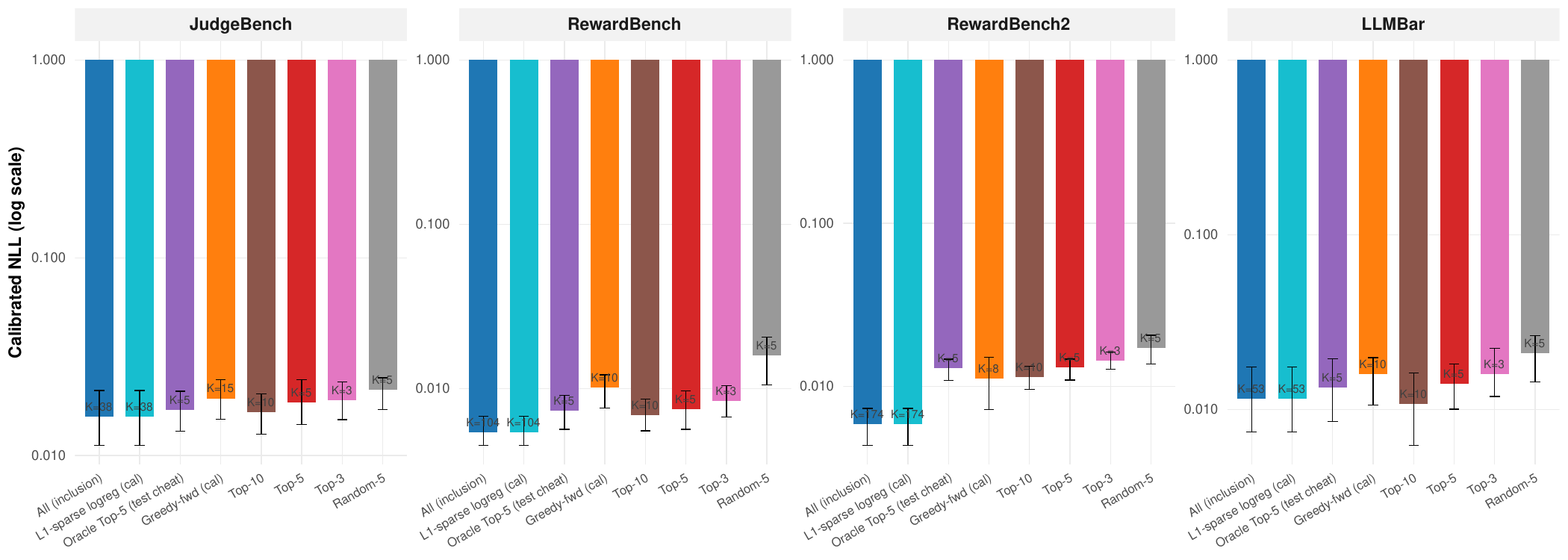}
\caption{Calibrated NLL for inclusion and curation baselines (lower is better).}
\label{fig:stronger-baselines}
\end{figure}

Inclusion beats greedy forward and oracle top-$5$ on all four datasets (all $p < 0.001$): JudgeBench ($\Delta{=}+0.0037$, $+0.0014$), RewardBench ($+0.0045$, $+0.0017$), RewardBench~2 ($+0.0052$, $+0.0071$), LLMBar ($+0.0043$, $+0.0018$).
Oracle top-5 confirms perfect accuracy ranking cannot rescue pruning at $k{=}5$; greedy forward overfits on the larger benchmarks; and $\ell_1$ logistic regression does not sparsify at any regularization strength, matching the full panel.
The only boundary remains LLMBar Top-10 ($\Delta = -0.0007$, $p = 0.033$), where two below-chance GPU judges (41\%, 45\%) approximate the $p \approx 1/2$ edge case of Theorem~1. Even there, neither greedy forward nor oracle top-5 beats inclusion.

A stricter test refits the aggregator+calibrator pipeline per candidate subset and selects by held-out calibrated NLL.\label{sec:same-pipeline-search}
We enumerate all $2^9{-}1=511$ subsets on the family-deduplicated JudgeBench panel, and use forward beam search (width $3$, up to $20$ steps) plus backward greedy elimination on the full panels.
Table~\ref{tab:same-pipeline} reports the chosen subset's test NLL; ``inclusion'' means non-overlapping $95\%$ CIs favor all judges, while ``tie'' means the CIs overlap.
Beam search loses to inclusion on every full panel ($\Delta \in [+0.0010,+0.0040]$).
Backward greedy ties inclusion but retains most judges on average ($K_{\mathrm{sel}}/K=26.5/38,\; 88.2/104,\; 158.1/174,\; 36.2/53$), making all-judges the simpler default.

\begin{table}[t]
\centering
\caption{Same-pipeline subset search versus all-judges Bayes+Beta. Forward beam loses to inclusion on every full panel; backward greedy and exhaustive family-dedup search tie inclusion.}
\label{tab:same-pipeline}
\small
\begin{tabular}{@{}l rr rr l@{}}
\toprule
Panel & $K$ & Search mode & Inclusion NLL & Search NLL & Winner \\
\midrule
JB-dedup (family) & $9$   & exhaustive       & $0.0188$ & $0.0189 \; (K_{\mathrm{sel}}{=}4.5)$   & tie \\
JudgeBench        & $38$  & beam (width 3)   & $0.0157$ & $0.0167 \; (K_{\mathrm{sel}}{=}16.2)$  & inclusion \\
JudgeBench        & $38$  & backward greedy  & $0.0157$ & $0.0157 \; (K_{\mathrm{sel}}{=}26.5)$  & tie \\
RewardBench       & $104$ & beam (width 3)   & $0.0057$ & $0.0095 \; (K_{\mathrm{sel}}{=}17.1)$  & inclusion \\
RewardBench       & $104$ & backward greedy  & $0.0057$ & $0.0055 \; (K_{\mathrm{sel}}{=}88.2)$  & tie \\
RewardBench~2     & $174$ & beam (width 3)   & $0.0057$ & $0.0107 \; (K_{\mathrm{sel}}{=}18.9)$  & inclusion \\
RewardBench~2     & $174$ & backward greedy  & $0.0057$ & $0.0058 \; (K_{\mathrm{sel}}{=}158.1)$ & tie \\
LLMBar            & $53$  & beam (width 3)   & $0.0115$ & $0.0124 \; (K_{\mathrm{sel}}{=}17.3)$  & inclusion \\
LLMBar            & $53$  & backward greedy  & $0.0115$ & $0.0114 \; (K_{\mathrm{sel}}{=}36.2)$  & tie \\
\bottomrule
\end{tabular}
\vspace{2pt}
\parbox{0.95\linewidth}{\footnotesize Winner: ``inclusion'' means non-overlapping $95\%$ CIs favor all judges; ``tie'' means the CIs overlap.}
\end{table}

\subsection{Calibrator Capacity Dominates Aggregator}
\label{sec:cross-ablation}
The previous comparisons use a fixed calibrated pipeline; we now isolate which part drives the NLL gains by running a full $\{\text{aggregator}\} \times \{\text{calibrator}\}$ grid with strict nested separation (70\% of calibration data for aggregator parameters, 30\% for calibrator parameters). Throughout the paper we use beta calibration with light elastic-net regularization ($\lambda{=}0.01$). Beta calibration brings every aggregator into the low-NLL regime ($\approx0.005$--$0.023$), whereas uncalibrated methods vary by two orders of magnitude ($0.20$--$1.88$). Aggregation is still useful: heterogeneity-aware methods beat vote fraction after calibration, but the main failure mode is raw-score miscalibration. Under the nested protocol, Platt scaling often collapses to the identity because its two parameters underfit the small calibrator sub-split; the full calibrator-design ablation is in Appendix~\ref{sec:calibrator-comparison}.

\begin{table}[t]
\centering
\caption{Aggregator $\times$ calibrator NLL under the nested 70/30 protocol (100 splits).}
\label{tab:cross-ablation}
\small
\begin{tabular}{@{}l ccc ccc@{}}
\toprule
& \multicolumn{3}{c}{No Calibration} & \multicolumn{3}{c}{Beta ($\lambda{=}0.01$)} \\
\cmidrule(lr){2-4} \cmidrule(lr){5-7}
Aggregator & JB & RB & RB2 & JB & RB & RB2 \\
\midrule
Majority Vote   & $0.568$ & $0.348$ & $0.629$ & $0.023$ & $0.020$ & $0.021$ \\
Accuracy-Wt     & $0.557$ & $0.313$ & $0.596$ & $0.023$ & $0.020$ & $0.020$ \\
Log-Odds        & $1.231$ & $0.415$ & $1.627$ & $0.016$ & $\mathbf{0.006}$ & $0.006$ \\
Dawid-Skene EM  & $1.682$ & $0.470$ & $1.876$ & $\mathbf{0.014}$ & $0.006$ & $\mathbf{0.005}$ \\
Bayes One-Coin  & $1.218$ & $0.411$ & $1.454$ & $0.016$ & $0.006$ & $0.006$ \\
Top-3 + Bayes   & $0.497$ & $0.200$ & $0.520$ & $0.019$ & $0.008$ & $0.014$ \\
Top-5 + Bayes   & $0.583$ & $0.214$ & $0.679$ & $0.019$ & $0.007$ & $0.013$ \\
Top-10 + Bayes  & $0.795$ & $0.257$ & $0.861$ & $0.017$ & $0.007$ & $0.011$ \\
\bottomrule
\end{tabular}
\end{table}

\subsection{Label Budget and Scope Diagnostics}
\label{sec:scope-and-budget}

We close the experiments by making the inclusion recommendation operational: how many calibration labels are needed, and when should added judges be filtered or handled by a more robust aggregator?

A natural concern is that inclusion may only beat selection when calibration labels are abundant. We sweep $n_\mathrm{cal}$ from $20$ to the full dataset (Figure~\ref{fig:label-budget}). Across all tested datasets and budgets, inclusion dominates selection, with the largest gains at small budgets: on RewardBench, all-judges reaches NLL $0.01$ with $75$ labels while top-$5$ needs more than $300$. Because this is a simultaneous across-budget claim, we test the curve with a Gaussian multiplier bootstrap~\citep{cck2013}; the observed max-$t$ statistic exceeds the $99.9\%$ null quantile by $4.1\times$, $6.1\times$, $23\times$, and $3.1\times$ on JudgeBench, RewardBench, RewardBench~2, and LLMBar, respectively (joint $p<10^{-4}$ on all four datasets; Appendix~\ref{app:joint-test}). The sharp drop near $n_\mathrm{cal}{\approx}30$ is calibrator-specific: Platt and temperature scaling lack this transition (Appendix~\ref{app:matched-cal}), so the effect reflects beta-calibrator identifiability, not a universal sample-size threshold.

\begin{figure}[t]
\centering
\includegraphics[width=\textwidth]{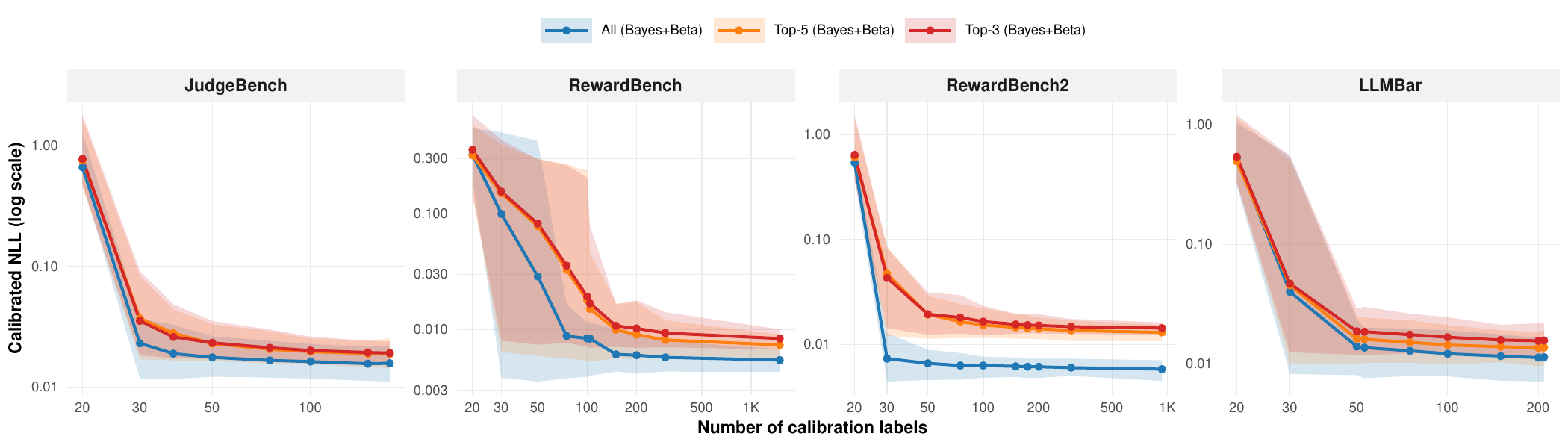}
\caption{Calibrated NLL vs.\ calibration-label budget; bands are 95\% CIs.}
\label{fig:label-budget}
\end{figure}

Inclusion is the default only when added judges provide calibratable signal.
We use three calibration-set diagnostics:
(P1)~\emph{parseability}, verdict parse rate ${\geq}90\%$ per judge;
(P2)~\emph{non-saturation}, calibrated NLL remains above the sharpness floor $\epsilon^\star\!\approx\!0.005$;
and (P3)~\emph{aggregator compatibility}, missingness pattern matches the aggregator.
These are conservative operating checks, all evaluated on the calibration set only.
JudgeBench and RewardBench satisfy all three.
LLMBar triggers P1 for one low-coverage GPU judge and has two below-chance judges, which explains the Top-$10$ boundary; Top-$3$/Top-$5$ comparisons remain intact.
RewardBench~2 under Bayesian one-coin violates P2/P3: $174$ related reward models already drive NLL near the floor, and four partially missing GPU judges perturb the aggregator, whereas Dawid--Skene absorbs them. Additional scope checks are reported in Appendix~\ref{app:scope-diagnostics}.



%% file: sections/related_work.tex
\section{Related Work}
\label{sec:related}

Classical multi-annotator aggregation begins with Dawid--Skene EM~\citep{dawid1979maximum}, which estimates annotator reliability from disagreement, and later crowdsourcing models extend this idea to item difficulty, labeller quality, and noisy pairwise rankings~\citep{whitehill2009whose,raykar2010learning,chen2013crowd}. Recent LLM-judge work studies richer unsupervised aggregation and uncertainty structure: \citet{snorkel2025frommanyvoices} propose \textsc{Care}, which models judge bias and correlation without labels, while \citet{wang2025noise} decompose evaluation noise and show that prediction noise is often dominant. Related Bradley--Terry variants model judge-specific discrimination or rater quality~\citep{xu2026judge,aczel2025bbq,btsigma2026}. We work in the labeled-calibration regime instead: the question is not only how to aggregate heterogeneous judges, but whether a limited labeled budget should calibrate all judges or curate a smaller panel.

A parallel line studies label-efficient evaluation with small amounts of ground truth. Prediction-Powered Ranking~\citep{chatzi2024prediction} uses a calibration set to form confidence regions for model rankings, \textsc{Precise}~\citep{divekar2026precise} reduces annotation cost with sub-instance estimates, \citet{dorner2025llm} characterize when LLM-judge assistance can or cannot reduce the required human-label budget, and Trust-or-Escalate~\citep{jung2025trust} uses calibrated confidence to decide when cheap judges should defer to expensive ones. Our label-budget question is complementary: given many imperfect judges and a fixed labeled calibration split, we find that spending the budget on post-hoc calibration of the full panel is usually better than accuracy-ranked pruning.

Our method uses standard post-hoc calibration and coverage tools rather than introducing a new calibrator. Step~3 builds on Platt scaling~\citep{platt1999probabilistic}, beta calibration~\citep{kull2017beta}, and the modern calibration literature~\citep{guo2017calibration}; NLL is a strictly proper scoring rule~\citep{gneiting2007strictly}. Step~4 optionally wraps the calibrated probabilities with split conformal prediction~\citep{tibshirani2019conformal}. A fuller review of crowdsourcing, Bradley--Terry variants, single-judge calibration, and conformal extensions is in Appendix~\ref{app:additional-related}.

%% file: sections/conclusion.tex
\section{Discussion and Conclusion}
\label{sec:conclusion}

We presented a theory-guided empirical study of calibrated multi-judge LLM evaluation. Monotonicity of oracle-calibrated risk (Proposition~\ref{prop:monotone-risk}) and the Calibrated Jury Theorem (Theorem~\ref{thm:calibrated-cjt}) predict that retaining judge signals should improve calibrated probability estimation even when some judges are weak or redundant. Across four labeled benchmarks, this prediction largely holds: the inclusion advantage is significant for the primary Top-3 and Top-5 comparisons on every dataset (Table~\ref{tab:topk-inference}), survives judge-family deduplication, and remains competitive against stronger subset search. In short: on labeled benchmarks with $n_\mathrm{cal} \ge 30$, accuracy-ranked pruning is dominated by all-judges calibration, unless a judge is pathological or the panel is already saturated.

The practical recommendation is to start with all judges, a heterogeneity-aware aggregator, and regularized beta calibration; then use calibration-set diagnostics to decide whether this default is safe. If added judges are mostly unparseable or reliably below chance, filter them; if a large panel is already at the calibrated NLL floor and new partially missing judges perturb Bayesian one-coin, switch to a more robust aggregator such as Dawid--Skene. These operating details, including the concrete thresholds and boundary cases, are collected in Appendix~\ref{app:scope-diagnostics}. The scope is deliberately limited: we study labeled benchmarks with held-out calibration labels, not open-ended evaluation without ground truth; the one-coin theory idealizes judge dependence; and the IRT variance decomposition is a diagnostic model rather than a complete account of evaluation uncertainty.

Calibrated evaluation uncertainty can improve deployment practice by identifying when model comparisons should not be trusted. It also shifts the budget tradeoff: in labeled evaluation regimes, retaining a diverse judge panel and investing in calibration data can be more valuable than curating a small panel of high-accuracy judges.

%% file: sections/appendix.tex
\clearpage
\makeatletter
\newcommand*{\l@appsection}{\@dottedtocline{1}{0em}{2.6em}}
\newcommand*{\l@appsubsection}{\@dottedtocline{2}{1.8em}{3.6em}}
\begingroup
\section*{LIST OF APPENDICES}
{\small
\let\l@section\l@appsection
\let\l@subsection\l@appsubsection
\@starttoc{app}
}
\endgroup
\clearpage

\let\origappsection\section
\let\origappsubsection\subsection
\renewcommand{\section}[1]{\origappsection{#1}\addcontentsline{app}{section}{\protect\numberline{\thesection}#1}}
\renewcommand{\subsection}[1]{\origappsubsection{#1}\addcontentsline{app}{subsection}{\protect\numberline{\thesubsection}#1}}
\makeatother

\section{Notation}
\label{app:notation}

Table~\ref{tab:notation} lists every symbol used in the main text, theory section, and experiments.
We use four kinds of probability function and keep them strictly distinct:
\begin{itemize}[leftmargin=*,itemsep=2pt]
    \item $\mathbb{P}(\cdot)$ — probability under the (unknown) data-generating distribution. Used to define what \emph{calibration} means: $\mathbb{P}(z_t = 1 \mid \hat{p}_t = p) \approx p$.
    \item $\hat{P}(\cdot)$ and $\hat{p}_t$ — the model's predicted probability (function form, and value at item $t$). These are random variables that depend on the calibration sample.
    \item $P(\cdot)$ — probability under a specified parametric model (e.g., the one-coin Bayesian posterior of Section~\ref{sec:step1}, where $P(z_t \mid \mathbf{y}_t)$ refers to the posterior under known judge-accuracy priors).
    \item $\Pr(\cdot)$ — probability under the oracle generative model used in Section~\ref{sec:theory} (one-coin model with deterministic accuracies $p_k$). Reserved for theory statements where the data-generating mechanism is fully specified.
\end{itemize}
Hat ($\hat{\cdot}$) marks model estimates throughout; subscripts $t, k, i, j$ index items and judges.

{\small
\setlength{\tabcolsep}{4pt}
\renewcommand{\arraystretch}{1.15}
\begin{longtable}{@{}p{3.4cm}p{10.1cm}@{}}
\caption{Notation used throughout the paper. Probability symbols distinguish data-generating, model-predicted, parametric, and oracle quantities.}
\label{tab:notation} \\
\toprule
\textbf{Symbol} & \textbf{Meaning} \\
\midrule
\endfirsthead

\multicolumn{2}{@{}l}{\emph{(Table~\ref{tab:notation} continued from previous page)}} \\
\toprule
\textbf{Symbol} & \textbf{Meaning} \\
\midrule
\endhead

\midrule
\multicolumn{2}{r@{}}{\emph{(continued on next page)}} \\
\endfoot

\bottomrule
\endlastfoot

\multicolumn{2}{@{}l}{\emph{Probability functions}} \\
$\mathbb{P}(\cdot)$ & Probability under the data-generating distribution (used for the calibration condition) \\
$\hat{P}(\cdot)$, $\hat{p}_t$ & Model's predicted probability (function form, and scalar value at item $t$); $\hat{p}_t \in [0,1]$ \\
$\hat{p}_t^{\mathrm{corr}}$, $\hat{p}_t^{\mathrm{final}}$ & Step-2 (residual-corrected) and Step-3 (final-calibrated) probabilities \\
$P(\cdot)$ & Probability under a specified parametric model (one-coin Bayesian posterior, Section~\ref{sec:step1}) \\
$\Pr(\cdot)$ & Probability under the oracle generative model (Section~\ref{sec:theory}) \\
$p$ & Generic probability value in $[0,1]$ \\
$\pi$ & Class prior $\Pr(Y = 1)$ \\
$\sigma(x)$ & Logistic sigmoid $1/(1+e^{-x})$ \\
\midrule
\multicolumn{2}{@{}l}{\emph{Indices and counts}} \\
$K$, $N$ & Number of judges; number of pairwise items \\
$k \in [K]$, $t \in [N]$ & Judge index; item index \\
$i, j$ & Item / judge index in IRT (Section~\ref{sec:decomposition}) \\
$J_t \subseteq [K]$ & Set of judges with a non-missing verdict on item $t$ \\
$S, T \subseteq [K]$ & Judge subsets (theory); $S \subseteq T$ in monotonicity arguments \\
$n_k$, $c_k$ & Calibration items judge $k$ votes on; judge $k$'s correct count \\
$n_{\mathrm{cal}}$ & Calibration set size (held-out labeled subset) \\
\midrule
\multicolumn{2}{@{}l}{\emph{Observed and latent variables}} \\
$z_t \in \{0,1\}$ & Ground-truth label for item $t$ \\
$y_{tk} \in \{0,1,-1\}$ & Judge $k$'s verdict on item $t$; $-1$ denotes missing \\
$Y \in \{0,1\}$ & Random ground truth (theory) \\
$X_k \in \{-1,+1\}$ & Judge $k$'s vote in $\pm 1$ encoding (theory) \\
$\mathbf{x}_t$, $\mathbf{y}_t$ & Item-feature vector; vector of judge verdicts on item $t$ \\
$\mathbf{X}_S$ & Vector of judge votes restricted to subset $S$ \\
\midrule
\multicolumn{2}{@{}l}{\emph{Judge-reliability parameters}} \\
$a_k$ & Judge $k$'s accuracy (Bayesian r.v.\ in $\mathcal{M}_\omega$) \\
$p_k$ & Judge $k$'s accuracy (deterministic; theory) \\
$\hat{a}_k$ & Plug-in point estimate of $a_k$ \\
$\delta_k$ & Excess accuracy $p_k - 1/2$ (Theorem~\ref{thm:finite}) \\
$\alpha_k, \beta_k$ & Beta-posterior parameters: $\alpha_k = c_k + 1$, $\beta_k = (n_k - c_k) + 1$ \\
$\ell_k$ & Judge $k$'s log-odds contribution under the integrated posterior \\
$w_k$, $\alpha_k^{\mathrm{th}}$ & Aggregator weight on judge $k$; optimal $\alpha_k^{\mathrm{th}} = \log\!\bigl(p_k/(1-p_k)\bigr)$ \\
\midrule
\multicolumn{2}{@{}l}{\emph{Calibration maps}} \\
$(a, b, c)$ & Calibrator parameters (Platt: $a, b$; beta calibration: $a, b, c$) \\
$g(p)$, $G(\cdot)$ & Beta-calibration map $g(p) = \sigma(a \log p + b \log(1-p) + c)$; generic distribution-calibration map \\
$\boldsymbol{\gamma}$ & Coefficients on item features in residual correction \\
$\lambda$, $\eta$ & Elastic-net regularization strength; clipping floor on $\hat{p}$ (Lemma~\ref{lem:btl-bounded}) \\
$\mathcal{M}_\omega \!\subset\! \mathcal{M}_{\omega,\delta} \!\subset\! \mathcal{M}_{\omega,\delta,G}$ & Nested model spaces (Bayesian aggregation $\subset$ + residual correction $\subset$ + distribution calibration) \\
\midrule
\multicolumn{2}{@{}l}{\emph{Conformal coverage (Step 4)}} \\
$s_t$ & Nonconformity score $1 - \hat{p}_t^{\mathrm{final}}(z_t)$ \\
$1 - \alpha$ & Target coverage level \\
$\hat{q}_{1-\alpha}$ & Empirical $(1-\alpha)$-quantile of calibration nonconformity scores \\
$\mathcal{C}(x_t)$ & Conformal prediction set at level $1-\alpha$ \\
\midrule
\multicolumn{2}{@{}l}{\emph{IRT variance decomposition}} \\
$\theta_j$, $\beta_i$, $\alpha_j$ & Judge $j$ ability; item $i$ difficulty; discrimination (2PL IRT) \\
\midrule
\multicolumn{2}{@{}l}{\emph{Risk and information theory}} \\
$L(z, \hat{p})$ & Bernoulli log-loss $-[z\log\hat{p} + (1-z)\log(1-\hat{p})]$ \\
$R_{\mathrm{NLL}}(S)$, $R_{\mathrm{Brier}}(S)$ & Calibrated NLL / Brier risk under judge subset $S$ \\
$H(Y \mid \mathbf{X}_S)$ & Conditional Shannon entropy \\
$I(Y; \mathbf{X}_T \mid \mathbf{X}_S)$ & Conditional mutual information \\
$\eta_S(\mathbf{x}_S)$ & Oracle-calibrated predictor $\Pr(Y=1 \mid \mathbf{X}_S = \mathbf{x}_S)$ \\
$\delta_K(\mathbf{x})$ & Bayes-optimal decision rule (Theorem~\ref{thm:calibrated-cjt}) \\
$C_k$ & Per-judge contribution to the Calibrated Jury bound: $-\tfrac{1}{2}\log(1-(2p_k-1)^2)$ \\
\midrule
\multicolumn{2}{@{}l}{\emph{Joint inference (Appendix~\ref{app:joint-test})}} \\
$\mathcal{T}=\{\tau_1,\ldots,\tau_m\}$ & Grid of label budgets tested simultaneously \\
$T_b(\tau)$, $\bar T(\tau)$ & Per-split NLL difference at budget $\tau$ in split $b$; mean across splits \\
$\widehat{\mathrm{Corr}}$ & Empirical correlation of centered $T_b$ across budgets \\
$M_{\mathrm{obs}}$, $c_{1-\alpha}$ & Observed test statistic; bootstrap critical value at level $1-\alpha$ \\
\end{longtable}
}

\section{Additional Related Literature}
\label{app:additional-related}

\paragraph{Classical crowdsourcing and ranking aggregation.}
Dawid--Skene EM~\citep{dawid1979maximum} is the classical latent-label model for estimating annotator reliability from repeated noisy labels, and it motivates the reliability-weighted view used throughout this paper.
Raykar et al.'s probabilistic labeller model~\citep{raykar2010learning} and Whitehill et al.'s quality-weighted voting~\citep{whitehill2009whose} extend this idea by jointly modeling item difficulty and annotator competence.
For pairwise settings, \citet{chen2013crowd} study ranking aggregation from noisy comparisons, connecting crowdsourcing reliability to preference-learning models.
These works establish that disagreement is signal rather than nuisance, but they do not study the post-hoc calibrated all-versus-top-$k$ inclusion question under a held-out labeled calibration budget.

\paragraph{Bradley--Terry and heterogeneous-judge variants.}
Several recent LLM-as-a-judge papers modify Bradley--Terry-style models to account for judge heterogeneity.
\citet{xu2026judge} introduce a judge-specific discriminator parameter and prove consistency for the resulting estimator.
\citet{aczel2025bbq} use Bayesian Bradley--Terry modeling with rater-quality priors, while Bradley--Terry-$\sigma$ (BT-$\sigma$)~\citep{btsigma2026} shows how ignoring heterogeneity can make confidence intervals too narrow around biased estimates.
These methods are closest to our aggregation layer, but our main empirical comparison asks whether reliability-aware aggregation should be applied to the full panel or to an accuracy-pruned subset after calibration.

\paragraph{Single-judge calibration and judge confidence.}
LLM judges are often overconfident~\citep{tian2025overconfidence}, motivating a growing line of single-judge calibration methods.
\citet{radharapu2025probes} train linear probes on hidden states with Brier-score supervision; \citet{manggala2024qacal} propose QA-calibration across question--answer subgroups; \citet{dadkhahi2025btd} calibrate Bradley--Terry--Davidson samples; \citet{quantllmjudges2025} align judge scores to human preferences with lightweight generalized linear model (GLM) / Bradley--Terry--Luce (BTL) regressors; and \citet{chen2026blackbox} derive black-box confidence from token-level entropy.
These works calibrate an individual judge or a fixed judge-output representation, whereas our setting must also decide how many heterogeneous judges to retain before calibration.

\paragraph{Calibration and conformal extensions.}
Our probability maps use standard post-hoc tools: Platt scaling~\citep{platt1999probabilistic}, beta calibration~\citep{kull2017beta}, and the modern neural-network calibration literature~\citep{guo2017calibration}; NLL is a strictly proper scoring rule~\citep{gneiting2007strictly}.
For distribution-free sets, split conformal prediction~\citep{tibshirani2019conformal,barber2023conformal} provides marginal coverage under exchangeability.
\citet{sesia2025adaptive} and \citet{einbinder2024labelnoise} study conformal prediction under label contamination, while \citet{vanderlaan2024selfcal,vanderlaan2025venn} combine Venn--Abers or Venn-style calibration with conformal guarantees.
These extensions are useful for downstream coverage, but the main contribution here is upstream: calibrated probability estimation from a multi-judge panel.

\paragraph{Broader evaluation and uncertainty frameworks.}
Stacked generalization~\citep{wolpert1992stacked} underlies the idea of learning a second-stage map from base predictions, and the Predictability--Computability--Stability framework~\citep{yu2020veridical} gives a broader lens for separating statistical and computational reliability.
\textsc{Clear}~\citep{azizi2025clear} separates reducible and irreducible uncertainty in regression; we use the same intuition only interpretively when discussing whether new judges add marginal information or whether a panel is saturated.
Prediction-Powered Ranking~\citep{chatzi2024prediction}, \textsc{Precise}~\citep{divekar2026precise}, \citet{dorner2025llm}, and Trust-or-Escalate~\citep{jung2025trust} all study labeled-budget evaluation, but focus on ranking or escalation rather than the inclusion-versus-curation decision for a fixed multi-judge panel.

\section{Proofs}
\label{app:proofs}

\subsection{Bayes-Optimal Log-Odds Aggregation}

\begin{proposition}[Bayes-optimal log-odds aggregation]
\label{prop:logodds}
Under the one-coin model with judge accuracies $a_k\in(0,1)\setminus\{1/2\}$ and conditional independence given $Y$, the posterior log-odds is
\begin{equation}
\log \frac{P(z_t=1\mid\mathbf{y}_t)}{P(z_t=0\mid\mathbf{y}_t)}
= \log\frac{\pi}{1-\pi} + \sum_{k=1}^K (2y_{tk}-1)\log\frac{a_k}{1-a_k}.
\end{equation}
\end{proposition}

\begin{proof}
Under the one-coin model, conditional on $z_t=1$ judge $k$ outputs $y_{tk}=1$ with probability $a_k$ and $y_{tk}=0$ with probability $1-a_k$; conditional on $z_t=0$ these probabilities are reversed. Bayes' rule and conditional independence give
\begin{align*}
\log \frac{P(z_t=1\mid\mathbf{y}_t)}{P(z_t=0\mid\mathbf{y}_t)}
&= \log\frac{\pi}{1-\pi} + \sum_{k=1}^K \log\frac{P(y_{tk}\mid z_t=1)}{P(y_{tk}\mid z_t=0)} \\
&= \log\frac{\pi}{1-\pi} + \sum_{k=1}^K \left[y_{tk}\log\frac{a_k}{1-a_k}+(1-y_{tk})\log\frac{1-a_k}{a_k}\right] \\
&= \log\frac{\pi}{1-\pi} + \sum_{k=1}^K (2y_{tk}-1)\log\frac{a_k}{1-a_k}.
\end{align*}
The identity extends by continuity to $a_k=1/2$, where judge $k$ receives zero weight.
\end{proof}

\subsection{Proof of Monotonicity of Oracle-Calibrated Risk}
\label{app:proof-monotone-risk}

\begin{proof}[Proof of Proposition~\ref{prop:monotone-risk}]
Since $S\subseteq T$, $\sigma(\mathbf{X}_S)\subseteq\sigma(\mathbf{X}_T)$. For Brier risk, the Bayes action is $\eta_S(\mathbf{X}_S)$ and the Bayes risk is $\mathbb{E}[\mathrm{Var}(Y\mid\mathbf{X}_S)]$. The conditional law of total variance gives
\[
\mathrm{Var}(Y\mid\mathbf{X}_S)=\mathbb{E}[\mathrm{Var}(Y\mid\mathbf{X}_T)\mid\mathbf{X}_S]+\mathrm{Var}(\mathbb{E}[Y\mid\mathbf{X}_T]\mid\mathbf{X}_S),
\]
and taking expectations yields $R_{\mathrm{Brier}}(T)\leq R_{\mathrm{Brier}}(S)$.
For NLL, the Bayes risk of the oracle conditional predictor is the conditional entropy. The entropy chain rule yields
\[
H(Y\mid\mathbf{X}_S)-H(Y\mid\mathbf{X}_T)=I(Y;\mathbf{X}_{T\setminus S}\mid\mathbf{X}_S)\geq0,
\]
with equality if and only if $Y\perp\!\!\!\perp \mathbf{X}_{T\setminus S}\mid\mathbf{X}_S$, equivalently the displayed conditional mutual information is zero. This proves both monotonicity and the stated strictness condition.
\end{proof}

\subsection{Proof of the Calibrated Jury Theorem}
\label{app:proof-calibrated-cjt}

\begin{proof}[Proof of Theorem~\ref{thm:calibrated-cjt}]
For a fixed vote vector $\mathbf{x}$, the likelihood ratio between $Y=+1$ and $Y=-1$ is
\[
\log\frac{\Pr(\mathbf{X}=\mathbf{x}\mid Y=+1)}{\Pr(\mathbf{X}=\mathbf{x}\mid Y=-1)}
=\sum_{k=1}^K x_k\log\frac{p_k}{1-p_k}=\sum_{k=1}^K\alpha_k x_k,
\]
so thresholding this log-likelihood ratio is the Bayes decision under equal priors. For one judge, the Bhattacharyya coefficient between the two class-conditional binary channels is
\begin{align*}
\sum_{x\in\{-1,+1\}}\sqrt{\Pr(X_k=x\mid Y=+1)\Pr(X_k=x\mid Y=-1)}
&=2\sqrt{p_k(1-p_k)} \\
&=\sqrt{1-(2p_k-1)^2}.
\end{align*}
Conditional independence makes the joint coefficient factorize, giving
\[
\rho_K=\prod_{k=1}^K\sqrt{1-(2p_k-1)^2}=\exp\!\left(-\sum_{k=1}^K C_k\right).
\]
The Bhattacharyya bound for equal priors gives $\Pr(\delta_K\neq Y)\leq\tfrac12\rho_K$. Finally, $C_k\geq\tfrac12(2p_k-1)^2$ follows from $-\log(1-u)\geq u$ with $u=(2p_k-1)^2$. Hence $\sum_k(2p_k-1)^2=\infty$ implies $\sum_k C_k=\infty$ and the bound tends to zero.
\end{proof}

\subsection{Technical Lemma and Proof}

\begin{lemma}[Bounded clipped log-loss]
\label{lem:btl-bounded}
Let $\hat{p}\in[\eta,1-\eta]$ for some clipping floor $\eta\in(0,1/2)$ and let
$L(z,\hat{p})=-[z\log\hat{p}+(1-z)\log(1-\hat{p})]$ be Bernoulli log-loss.
Then $0\leq L(z,\hat{p})\leq B$ with $B=|\log\eta|$.
The unrestricted Bernoulli log-loss is strictly proper, and on the clipped domain its conditional risk at true probability $q$ is minimized at $\Pi_{[\eta,1-\eta]}(q)$, uniquely at $q$ whenever $q\in[\eta,1-\eta]$.
\end{lemma}

\begin{proof}
If $z=1$, then $L(z,\hat{p})=-\log\hat{p}\leq-\log\eta$; if $z=0$, then $L(z,\hat{p})=-\log(1-\hat{p})\leq-\log\eta$. Nonnegativity is immediate.
For properness, the conditional risk at true Bernoulli probability $q$ is
\[
\ell_q(p)=-q\log p-(1-q)\log(1-p).
\]
On $(0,1)$, $\ell_q'(p)=(p-q)/(p(1-p))$ and $\ell_q''(p)=q/p^2+(1-q)/(1-p)^2>0$, so the unique unrestricted minimizer is $p=q$. Restricting $p$ to $[\eta,1-\eta]$ therefore moves the minimizer to the Euclidean projection of $q$ onto that interval, and leaves it at $q$ when $q$ already lies in the clipped domain.
\end{proof}

\subsection{Nested-Calibration Excess-Risk Bound}

\begin{theorem}[Calibrated excess-risk bound under nested model expansion]
\label{thm:calibration}
Let $\mathcal{F}_0\subseteq\mathcal{F}_1$ be two probabilistic forecaster classes mapping inputs to $[\eta,1-\eta]$, and let
$R(f)=\mathbb{E}[L(Y,f(X))]$ be clipped Bernoulli log-loss risk.
Let $f_0^\star\in\arg\min_{f\in\mathcal{F}_0}R(f)$.
If $\hat{f}_1\in\mathcal{F}_1$ is an $\varepsilon_n$-oracle fit in the expanded class, meaning
\[
R(\hat{f}_1)\leq \inf_{f\in\mathcal{F}_1}R(f)+\varepsilon_n,
\]
then
\[
R(\hat{f}_1)-R(f_0^\star)\leq\varepsilon_n.
\]
Consequently, if $\varepsilon_n\to0$, moving from the smaller calibrated model class to the nested larger class cannot increase asymptotic log-loss.
For empirical risk minimization, the oracle condition holds with $\varepsilon_n=2\Delta_n$, where $\Delta_n=\sup_{f\in\mathcal{F}_1}|R(f)-\hat{R}_n(f)|$.
\end{theorem}

\begin{proof}
By nesting, $f_0^\star\in\mathcal{F}_1$, so
\[
\inf_{f\in\mathcal{F}_1}R(f)\leq R(f_0^\star).
\]
Combining this with the $\varepsilon_n$-oracle assumption gives
\[
R(\hat{f}_1)\leq\inf_{f\in\mathcal{F}_1}R(f)+\varepsilon_n\leq R(f_0^\star)+\varepsilon_n,
\]
which is the claimed bound. If $\hat{f}_1$ minimizes empirical risk over $\mathcal{F}_1$, then for $f_1^\star\in\arg\min_{f\in\mathcal{F}_1}R(f)$,
\[
R(\hat{f}_1)\leq \hat{R}_n(\hat{f}_1)+\Delta_n\leq \hat{R}_n(f_1^\star)+\Delta_n\leq R(f_1^\star)+2\Delta_n,
\]
so the oracle condition holds with $\varepsilon_n=2\Delta_n$. Lemma~\ref{lem:btl-bounded} ensures the loss is bounded, making standard uniform-convergence bounds applicable to $\Delta_n$ for finite or suitably controlled calibrator classes.
\end{proof}

\subsection{Robustness to Monotone Distortions}

\begin{proposition}[Monotone-distortion robustness of calibrated risk]
\label{prop:monotone}
Let $R\in[\eta,1-\eta]$ be a raw score and let $g$ be a strictly monotone bijection from $[\eta,1-\eta]$ onto its image.
For any post-hoc calibration class $\mathcal{H}$, define the transformed-score class
$\mathcal{H}_g=\{h\circ g^{-1}:h\in\mathcal{H}\}$.
Then
\[
\inf_{h\in\mathcal{H}}\mathbb{E}[L(Y,h(R))]
=
\inf_{\tilde h\in\mathcal{H}_g}\mathbb{E}[L(Y,\tilde h(g(R)))] .
\]
Thus an invertible monotone pre-calibration distortion cannot change the oracle calibrated risk when the calibrator class is transformed along with the score; with $\varepsilon_n$-oracle fitting on either class, the two risks differ by at most the corresponding oracle errors.
\end{proposition}

\begin{proof}
Because $g$ is one-to-one, $R$ and $g(R)$ generate the same sigma-field. For every $h\in\mathcal{H}$, the map $\tilde h=h\circ g^{-1}$ satisfies $\tilde h(g(R))=h(R)$ almost surely, so the transformed class attains every risk value attained by the original class. Conversely, every $\tilde h\in\mathcal{H}_g$ has the form $h\circ g^{-1}$ for some $h\in\mathcal{H}$, so it attains no additional risk values. The infima are therefore equal. The final statement follows by adding the two oracle-error bounds from Theorem~\ref{thm:calibration}.
\end{proof}

\subsection{Finite-Sample Hoeffding Bound for Weighted Voting}

\begin{theorem}[Finite-sample improvement of weighted over majority voting]
\label{thm:finite}
Consider $K$ independent binary judges with accuracies $a_k=1/2+\delta_k$ and votes encoded as $2y_k-1\in\{-1,+1\}$.
Let $S_w=\sum_k w_k(2y_k-1)$ and let $S_{\mathrm{mv}}$ be the unweighted sum.
Hoeffding's inequality gives error exponents
\[
E_w=\frac{(\sum_k(2a_k-1)w_k)^2}{2\sum_k w_k^2},\qquad
E_{\mathrm{mv}}=\frac{(\sum_k(2a_k-1))^2}{2K}.
\]
For log-odds weights $w_k=\log(a_k/(1-a_k))$ and small $\max_k|\delta_k|$,
\[
E_w-E_{\mathrm{mv}}=2K\,\mathrm{Var}(\delta_1,\ldots,\delta_K)+O\!\left(\sum_{k=1}^K |\delta_k|^4\right).
\]
Thus the leading second-order gain is positive whenever judge accuracies are heterogeneous; for sufficiently small nonzero $\delta_k$, log-odds weighting improves the Hoeffding exponent over majority voting.
\end{theorem}

\begin{proof}
Conditional on $Y=1$, $X_k=w_k(2y_k-1)$ is independent, lies in $[-|w_k|,|w_k|]$, and has mean $(2a_k-1)w_k$.
Hoeffding's inequality applied to $S_w=\sum_kX_k$ gives
\[
\Pr(S_w\leq0\mid Y=1)\leq\exp\!\left(-\frac{(\sum_k(2a_k-1)w_k)^2}{2\sum_kw_k^2}\right)=\exp(-E_w),
\]
and the same calculation with $w_k\equiv1$ gives $E_{\mathrm{mv}}$.
Now write $a_k=1/2+\delta_k$. Since
\[
w_k=\log\frac{1+2\delta_k}{1-2\delta_k}=4\delta_k+O(\delta_k^3),
\]
we have
\[
\sum_k(2a_k-1)w_k=\sum_k 2\delta_k(4\delta_k+O(\delta_k^3))=8\sum_k\delta_k^2+O\!\left(\sum_k|\delta_k|^4\right)
\]
and
\[
\sum_kw_k^2=16\sum_k\delta_k^2+O\!\left(\sum_k|\delta_k|^4\right).
\]
Therefore $E_w=2\sum_k\delta_k^2+O(\sum_k|\delta_k|^4)$, while
$E_{\mathrm{mv}}=2(\sum_k\delta_k)^2/K=2K\bar\delta^2$.
Using $\sum_k\delta_k^2=K\bar\delta^2+K\mathrm{Var}(\delta_1,
\ldots,\delta_K)$ yields the displayed expansion.
\end{proof}

\section{Null Results on Preference Data}
\label{app:null-results}

On six subjective preference datasets with $10{+}$ exogenous comparison-level reliability signals, all effects are statistically equivalent to zero (two one-sided test (TOST) equivalence tests, margin $\pm 0.005$ NLL, all $p < 0.001$).

\begin{table}[h]
\centering
\caption{Six preference datasets where comparison-level reliability weighting is equivalent to zero under TOST (all $p<0.001$).}
\label{tab:null-appendix}
\small
\begin{tabular}{@{}lrrllr@{}}
\toprule
Dataset & $M$ & $n$ & Best exogenous signal & NLL diff (95\% CI) & TOST $p$ \\
\midrule
MT-Bench & 6 & 3.3K & Judge leave-one-class-out (LOCO) agreement & $+0.0002\ [-0.0001, +0.0004]$ & $<0.001$ \\
PRISM & 21 & 22K & English proficiency & $+0.0001\ [-0.0000, +0.0001]$ & $<0.001$ \\
Arena (orig.) & 63 & 40K & Pair consistency & Hurts & --- \\
UltraFeedback & 17 & 75K & Cross-dimension & Hurts & --- \\
Arena 140K & 53 & 98K & Language / fatigue & $+0.0001\ [+0.0000, +0.0001]$ & $<0.001$ \\
OASST1 & 1782 & 31K & Review count / language & $+0.0000\ [-0.0000, +0.0001]$ & $<0.001$ \\
\bottomrule
\end{tabular}
\end{table}

\section{Uncalibrated Selection Details}
\label{app:uncalibrated-selection}

Table~\ref{tab:uncalibrated} shows why pruning is tempting before calibration: on JudgeBench, Top-5 selection raises raw accuracy by about $5.3$ percentage points over the all-judge Bayesian one-coin aggregator, while the uncalibrated all-judge posterior is overconfident (NLL $>1.2$).
These uncalibrated gains motivate the calibrated comparison in Section~\ref{sec:selection-calibration}, where the same top-$k$ choices lose after beta calibration.

\section{Headline Inference Details}
\label{app:headline-inference}

For the inclusion-versus-top-$k$ comparisons in Table~\ref{tab:topk-inference}, the repeated $50/50$ calibration/evaluation splits share items and are not treated as independent experimental units.
We use bootstrap intervals over splits for effect-size uncertainty, and item-level paired permutation tests within each split as the primary significance test rather than split-level paired $t$-tests over split means.
The table reports the median paired-test statistic and $p$-value across the 100 splits; reporting the minimum or maximum across splits gives the same significance calls.

\section{Joint Test Across Label Budgets}
\label{app:joint-test}

The label-budget claim is \emph{simultaneous} across a grid of budgets $\mathcal{T} = \{\tau_1, \ldots, \tau_m\}$, so per-$\tau$ $p$-values are not the right inferential object.
We apply the Gaussian multiplier bootstrap of~\citet{cck2013}.

For each dataset, we run $B=200$ random $50/50$ calibration/test splits.
For split $b$ and budget $\tau \in \mathcal{T}$, we record the per-split difference
\[
T_b(\tau) = \mathrm{NLL}_{\mathrm{top}5}^{(b)}(\tau) - \mathrm{NLL}_{\mathrm{all}}^{(b)}(\tau).
\]
Under the null hypothesis $H_0: \mathbb{E}[T(\tau)] = 0 \ \forall \tau$ (inclusion and top-$5$ equal at every budget), the standardized vector
$\sqrt{B}\,\bar T(\tau) / \widehat{\mathrm{se}}(\tau)$ is approximately $\mathcal{N}(0, \widehat{\mathrm{Corr}})$, where $\widehat{\mathrm{Corr}}$ is the empirical correlation of the centered $T_b$.
We draw $10{,}000$ centered Gaussians $\xi^{(m)} \sim \mathcal{N}(0, \widehat{\mathrm{Corr}})$ and compute the null distribution of $M^{(m)} = \max_{\tau \in \mathcal{T}} |\xi^{(m)}(\tau)|$, then compare against the observed $M_{\mathrm{obs}} = \max_\tau |\sqrt{B}\,\bar T(\tau) / \widehat{\mathrm{se}}(\tau)|$.
This controls the familywise error rate over all budgets simultaneously.
Table~\ref{tab:joint-test} reports the results.

\begin{table}[h]
\centering
\caption{Gaussian multiplier joint test of inclusion versus top-$5$ across calibration budgets. All datasets reject the joint null at level $0.001$.}
\label{tab:joint-test}
\small
\begin{tabular}{@{}lrrrrrr@{}}
\toprule
Dataset & $|\mathcal{T}|$ & $M_{\mathrm{obs}}$ & $c_{0.95}$ & $c_{0.99}$ & $c_{0.999}$ & $p_{\mathrm{joint}}$ \\
\midrule
JudgeBench   & 8  & $15.21$ & $2.61$ & $3.11$ & $3.68$ & $<10^{-4}$ \\
RewardBench  & 9  & $23.09$ & $2.72$ & $3.21$ & $3.81$ & $<10^{-4}$ \\
RewardBench2 & 9  & $87.20$ & $2.72$ & $3.22$ & $3.80$ & $<10^{-4}$ \\
LLMBar       & 8  & $11.63$ & $2.67$ & $3.16$ & $3.73$ & $<10^{-4}$ \\
\bottomrule
\end{tabular}
\end{table}

The observed max-$t$ statistic exceeds the $99.9\%$ null quantile by $4.1\times$ on JudgeBench, $6.1\times$ on RewardBench, $23\times$ on RewardBench~2, and $3.1\times$ on LLMBar.
Because the grid has $|\mathcal{T}| \in \{8, 9\}$ moderately correlated budgets, the correlation-aware critical value $c_{0.95} \approx 2.6$--$2.7$ is tighter than the Bonferroni-corrected threshold $z_{1 - 0.025/9} \approx 2.98$; the test is both tighter and accounts for correlation, unlike a per-point Bonferroni correction.

The single exception is LLMBar at Top-10 (Table~\ref{tab:topk-inference}: $\Delta = -0.001$, $t = -2.7$, $p = 0.997$ wrong-sided), which corresponds to the scope boundary identified in Section~\ref{sec:family-clustering}: two GPU judges with below-chance accuracy (Gemma~2~9B at $41\%$, Mistral~7B at $45\%$) inject noise that beta calibration at $K{=}53$ cannot fully repair.
Even in this boundary case, Top-10 does not \emph{significantly} outperform inclusion---it merely ties.
Top-3 and Top-5 on LLMBar remain strongly significant ($t = 16.6$ and $9.8$, both $p < 0.001$), indicating that the boundary is narrow: only the specific combination of below-chance judges + moderate $K$ + liberal selection threshold ($k{=}10$, retaining $19\%$ of the panel) produces a non-significant reversal.
The exception is identified \emph{a priori} by the scope heuristic, and even in this worst case the effect size is negligible ($\Delta = 0.0007$ NLL).

\section{Additional Experiments}
\label{app:additional}

\subsection{Judge Panels, Augmentation, and Inference Protocol}
\label{app:judge-list}

\begin{figure*}[t]
\centering
\resizebox{\textwidth}{!}{%
\begin{tikzpicture}[
    >=Latex,
    font=\small,
    box/.style={
        draw,
        rounded corners=3pt,
        thick,
        align=center,
        minimum height=1.55cm,
        text=textdark
    },
    inputbox/.style={box, draw=inblueborder, fill=inblue, minimum width=3.3cm},
    prepbox/.style={box, draw=prepborder, fill=prepfill, minimum width=3.3cm},
    modelbox/.style={box, draw=modelborder, fill=modelfill, minimum width=3.5cm},
    parserbox/.style={box, draw=parserborder, fill=parserfill, minimum width=13.8cm, minimum height=2.1cm},
    outgood/.style={box, draw=modelborder, fill=modelfill, minimum width=5.5cm, minimum height=1.7cm},
    outbad/.style={box, draw=neutralborder, fill=neutralfill, minimum width=5.5cm, minimum height=1.7cm},
    legendbox/.style={draw, rounded corners=2pt, minimum width=0.6cm, minimum height=0.45cm},
    arrow/.style={-Latex, thick},
    downarrow/.style={-Latex, thick}
]

\node[inputbox] (pair) at (0,0)
{Pair item\\
\footnotesize (instruction + A, B)};

\node[prepbox] (rand) at (4.2,0)
{Randomize\\
A/B order};

\node[prepbox] (prompt) at (8.6,0)
{MT-bench pairwise\\
prompt $\rightarrow$ ``[[A]]/[[B]]''};

\node[modelbox] (backend) at (13.1,0)
{Judge backend\\
\footnotesize API ($T{=}0$) / GPU ($T{=}0.01$)};

\node[inputbox] (raw) at (17.6,0)
{Raw text\\
\footnotesize(strip \texttt{<think>...\ </think>})};

\draw[arrow] (pair.east) -- (rand.west);
\draw[arrow] (rand.east) -- (prompt.west);
\draw[arrow] (prompt.east) -- (backend.west);
\draw[arrow] (backend.east) -- (raw.west);

\node[parserbox] (parser) at (8.8,-3.2)
{\textbf{Four-level parser}\\[2pt]
\footnotesize
(1) \texttt{[[A]]} / \texttt{[[B]]}
\quad $\rightarrow$ \quad
(2) \texttt{**A**} / \texttt{**B**}
\quad $\rightarrow$ \quad
(3) trailing \texttt{Output/Solution/Assistant A|B}
\quad $\rightarrow$ \quad
(4) bare trailing \texttt{A|B}
};

\draw[downarrow] (raw.south) -- ++(0,-0.7) -| ($(parser.north)+(4.3,0)$);

\node[outgood] (good) at (5.2,-6.3)
{Verdict $\in \{\texttt{A>B}, \texttt{B>A}\}$\\
$\rightarrow$ compare to randomized GT\\
\footnotesize(correct / wrong)};

\node[outbad] (bad) at (12.8,-6.3)
{No match $\rightarrow$ missing ($-1$)\\
\footnotesize judge skipped for this item};

\draw[downarrow] ($(parser.south)+(-3.0,0)$) -- (good.north);
\draw[downarrow] ($(parser.south)+(3.0,0)$) -- (bad.north);

\node[legendbox, draw=inblueborder, fill=inblue] (l1) at (0.2,-8.1) {};
\node[anchor=west] at (0.55,-8.1) {\footnotesize Input/data};

\node[legendbox, draw=prepborder, fill=prepfill] (l2) at (3.3,-8.1) {};
\node[anchor=west] at (3.65,-8.1) {\footnotesize Preprocessing};

\node[legendbox, draw=modelborder, fill=modelfill] (l3) at (6.9,-8.1) {};
\node[anchor=west] at (7.25,-8.1) {\footnotesize Model};

\node[legendbox, draw=parserborder, fill=parserfill] (l4) at (9.5,-8.1) {};
\node[anchor=west] at (9.85,-8.1) {\footnotesize Parser};

\end{tikzpicture}
}
\caption{Inference pipeline for augmented judges. Randomized A/B prompts are scored by API or GPU judges, parsed by a four-level fallback parser, and unparseable outputs are marked missing.}
\label{fig:appendix_inference_pipeline}
\end{figure*}
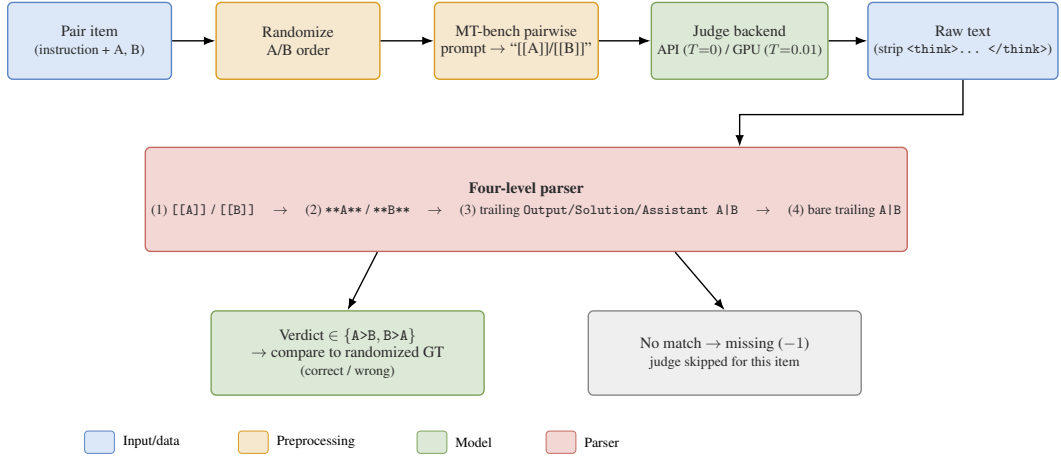


This subsection documents the judge panels, the rationale for augmenting the native benchmarks with additional LLM-as-a-judge systems, and the inference pipeline (Figure~\ref{fig:appendix_inference_pipeline}).

Two of the four benchmarks, RewardBench~\citep{lambert2024rewardbench} and RewardBench~2~\citep{malik2025rewardbench2}, ship with enough judges ($100$ and $174$ reward models) that the accuracy distribution is already wide and the aggregator has plenty of signal.
The other two are different.
JudgeBench~\citep{judgebench2025} bundles verdicts from $32$ native judges concentrated in a narrow $54\%$--$84\%$ accuracy band, and LLMBar~\citep{zeng2024llmbar} ships with $44$ evaluator records that come from a small number of backbone models combined with several prompt strategies, so even ``different'' LLMBar judges are often the same base model under different instructions, which undercuts one of our claims (that \emph{heterogeneous} weak judges carry material calibration signal).
To stress-test that claim we extended both panels with a mix of API-hosted models and local GPU-hosted open-weight models that span a wider accuracy range and a wider set of backbone families.
Concretely, we added $6$ API judges to JudgeBench and $5$ API $+$ $4$ GPU judges to LLMBar; for the saturation analysis we further added the same $4$ GPU judges to both JudgeBench and RewardBench~2.
Because RewardBench on its own already has a wide, redundancy-free panel, we did not augment it.

All added judges share an identical MT-bench-style pairwise prompt~\citep{zheng2023judging}: the judge is told it is an ``impartial judge,'' given the user instruction and the two candidate responses (as ``Assistant~A'' and ``Assistant~B''), and asked to produce a short justification followed by a strict verdict token, either \verb+[[A]]+ or \verb+[[B]]+.
For JudgeBench's coding subset we substitute an analogous prompt that shows the two code solutions instead of free-form responses.
To remove position bias, we randomize which of the two responses is labelled Assistant~A and which is Assistant~B \emph{independently per item}; a verdict is then scored as correct when it matches the randomized-ground-truth side.
This sidesteps the need to query each judge twice (with and without a swap) and keeps the prediction matrix tractable.

For all API backends (OpenAI, Anthropic, and Groq, which we use as a fast host for open-weight API judges) we set temperature to $0$ so the verdict is deterministic conditional on the random A/B permutation.
GPU judges (4-bit quantization, near-greedy sampling) use temperature $0.01$, the smallest non-zero value that keeps the HuggingFace generation pipeline in sampling mode and avoids a deterministic decoding failure we observed with temperature exactly $0$ on some checkpoints.
Maximum output length is $400$ tokens for OpenAI/Groq, $1500$ for Anthropic (its responses are reliably longer because of its preamble), $256$ for small CPU models, and $512$ for GPU models.

A single judge response is often longer and messier than a clean \verb+[[A]]+ token. Different backbones favour different formats, and some (e.g.\ Qwen3) emit \verb+<think>+\ldots\verb+</think>+ reasoning traces before the final answer.
We strip any reasoning block and then run a four-level parser that tries progressively looser matches:
(i) the strict bracketed token \verb+[[A]]+ or \verb+[[B]]+;
(ii) markdown-bold \verb+**A**+ or \verb+**B**+;
(iii) trailing ``Output~A'' / ``Solution~A'' / ``Assistant~A'' (or the $B$ variants) in the final $80$ characters of the response;
(iv) a bare $A$ or $B$ as the last non-whitespace character.
The parser stops at the first level that matches; if none matches, the verdict is recorded as \emph{missing} ($-1$ in the prediction matrix) and the aggregator simply skips that judge for that item.
Levels~(ii)--(iv) cover common failure modes we observed empirically; for example, Llama~3.1~8B under some prompts produces a bolded verdict without brackets, and several smaller models end with an unadorned ``A.'' or ``B.'' on a final line.

Table~\ref{tab:judge-panel} gives per-judge, per-dataset verdict counts.
For the four $7$B--$9$B GPU judges, the parser recovers $92\%$--$100\%$ of items on JudgeBench and RewardBench~2; on LLMBar, three of the four GPU judges cover $99\%$--$100\%$ of items, while the Llama~3.1~8B judge produces $54\%$ parseable verdicts even after a re-run with a revised prompt template.
This last judge is the \emph{single} concrete case that triggers the parseability condition~(P1) in our operational scope heuristic (Section~\ref{par:scope-failures}); we keep it in the main panel to document that the heuristic catches real-world noise rather than a hypothetical failure mode.
API judges return either a complete response or a clearly-marked error, and do not contribute additional parseability failures.

\begin{table}[h]
\centering
\caption{Judge panel composition and parsed-verdict coverage. Main panels use $K{=}38$ (JB), $174$ (RB2), and $53$ (LLMBar); daggers mark runs reported but not included.}
\label{tab:judge-panel}
\small
\begin{tabular}{@{}l l rrr@{}}
\toprule
Source & Judge (model) & JudgeBench & RB2 & LLMBar \\
\midrule
\multirow{6}{*}{API}
 & GPT-4o-mini                          & $309/350$ & $1800/1865^{\dagger}$ & $\approx 100\%$ \\
 & Claude~3 Haiku                       & $164/350$ & $1130/1865^{\dagger}$ & $\approx 100\%$ \\
 & Qwen3 32B                            & $197/350$ & $1834/1865^{\dagger}$ & $\approx 100\%$ \\
 & Llama~3.1 8B Instant                 & $195/350$ & $1796/1865^{\dagger}$ & $\approx 100\%$ \\
 & Llama~3.3 70B Versatile              & $249/350$ & $1809/1865^{\dagger}$ & $\approx 100\%$ \\
 & Llama~4 Scout 17B (16e)              & $191/350$ & $1695/1865^{\dagger}$ & $\approx 100\%$ \\
\midrule
\multirow{4}{*}{GPU (7B--9B)}
 & Qwen~2.5 7B-Instruct                 & $348/350$ & $1723/1865$ & $419/419$ \\
 & Gemma~2 9B-it                        & $349/350$ & $1759/1865$ & $417/419$ \\
 & Llama~3.1 8B-Instruct                & $315/350$ & $1710/1865$ & $227/419$ \\
 & Mistral~7B-Instruct-v0.3             & $230/350$ & $1753/1865$ & $419/419$ \\
\bottomrule
\end{tabular}
\end{table}

The $\dagger$-marked RB2 entries reflect a separate set of API judge runs we collected on the same $1865$ items via OpenAI, Anthropic, and Groq endpoints; their accuracies on the items they cover are $0.79$ (GPT-4o-mini), $0.69$ (Claude~3 Haiku), $0.87$ (Qwen3~32B), $0.71$ (Llama~3.1 8B Instant), $0.83$ (Llama~3.3 70B Versatile), and $0.83$ (Llama~4 Scout~17B), all comparable to or above the median accuracy of the $174$ native reward models.
We keep the main RB2 panel at $K{=}174$ (or $K{=}178$ in the GPU saturation analysis) so that the saturation analysis in Section~\ref{sec:saturation} is not confounded by adding heterogeneous LLM-judge signal.
Folding these API judges into RB2 to test whether heterogeneous additions can break out of the redundancy floor is a natural extension, separate from the GPU experiment we already report.

On JudgeBench we retain the $32$ native judges and add $6$ API judges for the main panel ($K{=}38$); the four GPU judges are reserved for the saturation analysis because their lower coverage (65--100\% of items) and below-median accuracy would conflate the inclusion-vs-selection comparison with a coverage-missingness confound.
On RewardBench the main panel is the $100$ native reward models augmented with the four GPU judges ($K{=}104$).
On RewardBench~2 the main panel is the $174$ native reward models; the saturation experiment adds the four GPU judges to give $K{=}178$.
On LLMBar the main panel is $44$ native evaluators (multiple backbones $\times$ prompt strategies from the benchmark release) augmented with $5$ API judges and $4$ GPU judges for $K{=}53$; the single low-coverage entry is the GPU Llama~3.1 8B judge at $227/419$, which is the concrete case that triggers the parseability condition (P1) in the operational scope heuristic.

\subsection{Curation Penalty and Calibrated NLL Trajectory}
\label{app:teaser}

Figure~\ref{fig:teaser} plots the calibrated NLL as a function of the retained panel size for all four benchmarks.
On every dataset, top-$3$ selection incurs $1.2\times$--$2.5\times$ higher NLL than the full panel, and the trajectory generally improves as more judges are retained.
The curves complement the main-text violin plot by showing the full $k$-dependent trajectory rather than a single endpoint comparison.

\begin{figure}[h]
\centering
\includegraphics[width=0.95\textwidth]{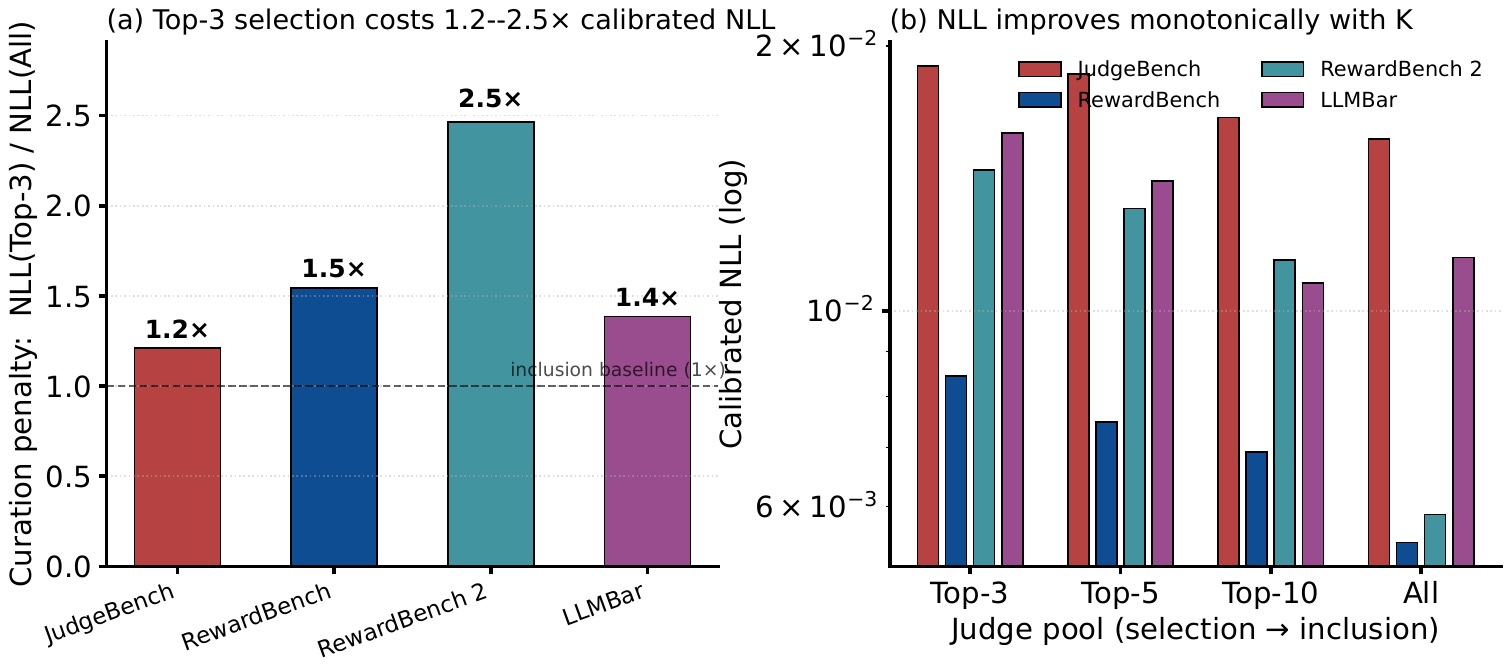}
\caption{Curation penalty and calibrated NLL trajectory across four benchmarks. Top-$3$ selection costs $1.2\times$--$2.5\times$ more NLL, and calibrated NLL generally improves as the retained panel grows.}
\label{fig:teaser}
\end{figure}

\subsection{Full NLL-vs-$k$ Curves (Inclusion vs.\ Selection)}
\label{app:inclusion-curves}

Figure~\ref{fig:inclusion-curves} shows the underlying NLL-vs-$k$ trajectories summarized by the main-paper violin+arrow plot.
Each panel plots calibrated and uncalibrated NLL as a function of the retained pool fraction $k/K_\mathrm{max}$, with error bars from $100$ random $50/50$ calibration/test splits.
The main-paper regime flip is the signed difference between the red and blue curves: red curves favor small selected panels before calibration, while blue curves generally favor retaining the full panel after calibration.
The main text uses the compact violin/arrow view; the full curves here show the trajectory at intermediate $k$ values.

\begin{figure}[h]
\centering
\includegraphics[width=\textwidth]{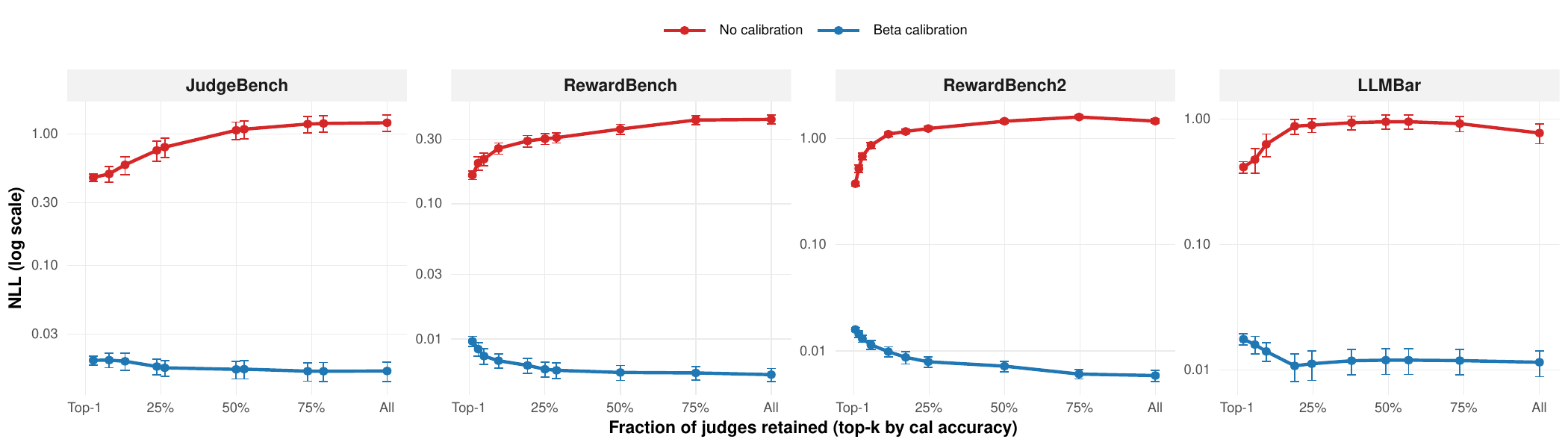}
\caption{Full NLL-vs-$k$ curves with and without beta calibration. Selection helps before calibration, while beta calibration generally shifts the optimum toward the full panel.}
\label{fig:inclusion-curves}
\end{figure}

\subsection{Stronger Curator Details}
\label{app:curator-details}

Greedy forward selection adds judges one at a time to minimize calibrated NLL on an inner calibration split, with $k_{\max}\in\{8,10,15\}$ chosen according to panel size.
It is calibration-aware but can overfit on larger panels; for example, on RewardBench greedy reaches NLL $0.0102$, worse than the simple top-$10$ baseline at $0.0069$.
$\ell_1$-sparse logistic regression fits a learned sparse combination on the raw vote matrix; across $C\in\{10^{-3},10^{-2},10^{-1},1\}$ it refuses to sparsify and effectively recovers the full panel.
Oracle top-$5$ ranks judges by test-set accuracy and is included only as a diagnostic: it asks whether perfect accuracy ranking at fixed $k=5$ would rescue pruning.
It is not a deployable method, because a calibration-aware curator could in principle do better.
Random-$5$ is a sanity floor for small selected panels.

For the same-pipeline subset search in Table~\ref{tab:same-pipeline}, each candidate subset $S$ refits the full Bayesian one-coin aggregator on the outer calibration split and beta calibration on an inner split, then selects the subset by inner-split calibrated NLL and reports NLL on the untouched test split.
On the family-deduplicated JudgeBench panel ($K=9$) we enumerate all $2^9-1=511$ nonempty subsets.
On the full panels, exhaustive search is infeasible, so we use forward beam search (width $3$, up to $20$ additions) and backward greedy elimination as complementary approximations.
This protocol is stricter than ordinary top-$k$ pruning because the subset is chosen after refitting the same aggregator+calibrator pipeline used by the all-judges arm.
Questions about whether the aggregation step itself matters are separated from the curator comparison and evaluated in Appendix~\ref{app:aggregator-comparison}.

\subsection{IRT Variance Decomposition and the Saturation Mechanism}
\label{app:epi-alea}

We fit a 2PL IRT model $P(\mathrm{correct} \mid j, i) = \sigma(\alpha_j(\theta_j - \beta_i))$ on each of the four datasets and decompose the total logit-scale variance into a judge-heterogeneity share $\mathbb{E}[\alpha^2]\mathrm{Var}(\theta)$ and an item-difficulty share $\mathbb{E}[\alpha^2]\mathrm{Var}(\beta)$ following~\citet{lord1968irt}.
The judge share is \emph{dominant on every dataset}: JudgeBench $84.3\%$, RewardBench $99.8\%$, RewardBench~2 $96.4\%$, LLMBar $94.2\%$.
This rules out the simplistic reading that RewardBench~2 saturates because item-intrinsic (aleatoric) variance has taken over; it has not.
What has happened is that the $174$ RewardBench~2 reward models come from a small number of closely related training recipes (shared base models, preference data, and evaluation styles), so the \emph{marginal} information a new judge carries given the remaining $173$ is already small, not that item difficulty has consumed the variance budget.
We interpret the RewardBench~2 boundary as a \emph{redundancy}-induced saturation, not an aleatoric one, and report the full IRT split in text rather than in a separate schematic.

The IRT $\theta$ ordering is also nearly identical to the simple raw-accuracy ordering: $\text{Corr}(\theta, \text{raw accuracy}) = 0.995$ on JudgeBench, confirming that the accuracy-based judge ranking used by top-$k$ pruning recovers the IRT-optimal ordering.
The dominant source of evaluation uncertainty is \emph{which judges you have}, not \emph{which items you evaluate}, which motivates the per-judge modeling that the main pipeline performs.

Separately, JudgeBench's position-swapped verdicts show that reward models achieve $99$--$100\%$ self-consistency, while LLM judges range from $40\%$ (JudgeLM-7B) to $85\%$ (o1-preview).
Position bias is the dominant within-judge noise source for LLM judges, and we mitigate it by randomizing the A/B order per item (Sec.~\ref{app:judge-list}).

\subsection{Synthetic Narrow-Band Ablation}
\label{app:narrow-band}

Theorem~\ref{thm:finite} predicts that the Hoeffding-exponent advantage of log-odds weighting over majority voting scales with the cross-judge accuracy variance $\mathrm{Var}(\delta_1,\ldots,\delta_K)$.
We test this on synthetic panels ($K{=}30$ judges, mean accuracy $0.70$, $N{=}2000$ items, 100 replicates) by varying the accuracy standard deviation from $0.01$ to $0.30$ and comparing log-odds weighting against top-half and top-third pruning on point accuracy.
When judges are nearly homogeneous ($\sigma{=}0.01$), pruning discards useful votes for no gain and log-odds weighting leads the best pruning rule by $3.2$ percentage points.
As heterogeneity grows, both methods approach perfect accuracy and the gap shrinks to below $0.01$ points at $\sigma \ge 0.15$.
The pattern is consistent with the theorem: weighting is most valuable precisely when judges are similar and pruning is most wasteful, while at high heterogeneity both approaches saturate.

\subsection{Matched-Calibrator Budget Sweep}
\label{app:matched-cal}

To verify that the $n_\mathrm{cal}\approx 30$ phase transition reported in Section~\ref{sec:label-budget} reflects the beta calibrator specifically, we repeated the budget sweep with two simpler alternatives, Platt scaling (2 parameters) and temperature scaling (1 parameter), while holding the aggregator (Bayesian one-coin on the full panel) fixed.
The calibrator is the only change.

\begin{table}[h]
\centering
\caption{Matched-calibrator budget sweep on the full judge panel. Beta calibration shows the sharp NLL drop near $n_\mathrm{cal}{=}30$; Platt and temperature remain nearly flat.}
\label{tab:matched-cal}
\small
\begin{tabular}{@{}l c rrrrr@{}}
\toprule
Dataset & Calibrator & $n{=}20$ & $n{=}30$ & $n{=}50$ & $n{=}100$ & $n{=}150$ \\
\midrule
\multirow{3}{*}{JudgeBench}
 & beta        & $0.567$ & $\mathbf{0.018}$ & $0.017$ & $0.016$ & $0.016$ \\
 & Platt       & $1.230$ & $1.230$ & $1.230$ & $1.230$ & $1.230$ \\
 & temperature & $0.567$ & $0.561$ & $0.552$ & $0.549$ & $0.547$ \\
\midrule
\multirow{3}{*}{RewardBench}
 & beta        & $0.288$ & $\mathbf{0.049}$ & $0.007$ & $0.006$ & $0.006$ \\
 & Platt       & $0.413$ & $0.413$ & $0.413$ & $0.413$ & $0.413$ \\
 & temperature & $0.288$ & $0.263$ & $0.241$ & $0.236$ & $0.233$ \\
\midrule
\multirow{3}{*}{RewardBench~2}
 & beta        & $0.558$ & $\mathbf{0.006}$ & $0.006$ & $0.006$ & $0.006$ \\
 & Platt       & $1.445$ & $1.445$ & $1.445$ & $1.445$ & $1.445$ \\
 & temperature & $0.558$ & $0.536$ & $0.528$ & $0.524$ & $0.521$ \\
\midrule
\multirow{3}{*}{LLMBar}
 & beta        & $0.431$ & $\mathbf{0.022}$ & $0.012$ & $0.011$ & $0.011$ \\
 & Platt       & $0.768$ & $0.768$ & $0.768$ & $0.768$ & $0.768$ \\
 & temperature & $0.431$ & $0.400$ & $0.387$ & $0.380$ & $0.376$ \\
\bottomrule
\end{tabular}
\end{table}

The pattern is consistent across all four benchmarks: beta calibration collapses NLL by roughly an order of magnitude once $n_\mathrm{cal}\ge 30$, whereas Platt scaling is frozen by its intercept-plus-slope parameterization and temperature scaling cannot shift the overall scale of the Bayesian one-coin output.
The transition is better read as ``when the beta calibrator's three parameters become identifiable'' than as a population-level statement about epistemic uncertainty.

\subsection{Additional Scope Diagnostics}
\label{app:scope-diagnostics}

We evaluate calibration requirements by sweeping the calibration fraction from 5\% to 70\% over 50 random splits.
Below 20\%, calibration reduces NLL by 40--54\% but leaves point accuracy nearly unchanged.
At 30\%+ (${\sim}100$ calibration items), NLL drops by 93\% and accuracy reaches 99\%+; performance is stable from 30\%--70\% (NLL 0.076--0.087).
Grouped difficulty splits reveal the expected asymmetry: calibrating on hard items generalizes well to easy items (NLL $4.92\to0.075$ with Platt), whereas calibrating on easy items only partially transfers to hard items (NLL $3.49\to0.49$).
This is the support-gap case where $Y\not\!\perp D\mid S$: hard and easy items with the same raw score need not share the same empirical accuracy.
Difficulty-aware calibration, stratified Platt, and importance-weighted Platt reduce but do not remove the gap, so the operational recommendation is to use random calibration splits covering the full difficulty range.

Calibration rescales confidence without changing the answer ordering: the Spearman rank correlation between raw and calibrated scores is $0.94\pm0.01$.
This addresses the concern that low NLL is obtained by changing the underlying preference decisions rather than by correcting probabilities.

The main-text operating rule uses only calibration-set quantities.
A judge $k$ should be reviewed if its parsed verdict rate is below $90\%$, and filtered if the failure is severe (around $70\%$ or lower) or if its calibration-set accuracy is reliably below chance: $\hat a_k < 0.5 - 2\widehat{\mathrm{se}}_k$, where $\widehat{\mathrm{se}}_k = \sqrt{\hat a_k(1-\hat a_k)/n_k}$.
When adding judges to a large panel ($K\gtrsim100$) changes calibrated NLL by only $\pm1\%$ around the floor, we treat the pool as saturated and prefer either filtering the offending judges or switching from Bayesian one-coin to Dawid--Skene.

The scope diagnostics in Section~\ref{par:scope-failures} correspond to three empirical boundary cases.
First, saturation appears when the existing pool is already at the calibrator's sharpness floor: on RewardBench~2, Bayesian one-coin with $174$ full-coverage reward models has calibrated NLL $\approx0.006$, and four $7$B--$9$B GPU judges with 5--8\% missing verdicts slightly raise NLL ($p<10^{-3}$).
The same GPU judges help under Dawid--Skene ($p=0.002$), so the boundary is aggregator-conditional rather than absolute.
Second, parser breakdown turns disagreement into measurement noise: on LLMBar, the Llama-3.1 GPU judge retains only 54\% valid verdicts after parser tuning, while the other three GPU judges recover 99--100\% coverage.
Third, aggregator mismatch matters because Bayesian one-coin is less robust to missing entries in an otherwise full-coverage pool than EM-estimated Dawid--Skene weights.

Finally, on Chatbot Arena with 63 models, calibrated log-odds aggregation recovers clean NLL under 40\% controlled corruption.
Across seven distortion types ($r^2$, $\sqrt{r}$, noisy $r+\mathcal{N}(0,\sigma^2)$ at two noise levels, binary noise, inversion, and random shuffle), calibrated methods remain well-calibrated; the monotone cases match Proposition~\ref{prop:monotone}.

\subsection{Does the Aggregator Step Matter?}
\label{app:aggregator-comparison}

We compare four aggregators on the same pipeline---raw vote fraction, the Bayesian one-coin model of Proposition~\ref{prop:logodds}, Dawid--Skene EM, and logistic stacking with out-of-fold predictions---each followed by beta calibration, on all four datasets with 100 random 50/50 splits.

\begin{center}
\small
\begin{tabular}{@{}l cccc@{}}
\toprule
Method & JudgeBench & RewardBench & RewardBench~2 & LLMBar \\
\midrule
Vote fraction + Beta   & $0.022$ & $0.020$ & $0.021$ & $0.023$ \\
\textbf{Bayes-1Coin + Beta} & $\mathbf{0.016}$ & $\mathbf{0.006}$ & $\mathbf{0.006}$ & $\mathbf{0.011}$ \\
\textbf{Dawid--Skene + Beta} & $\mathbf{0.013}$ & $\mathbf{0.006}$ & $\mathbf{0.005}$ & $\mathbf{0.009}$ \\
Stacking-CV + Beta$^\dagger$     & $0.001$ & $0.001$ & $0.001$ & $0.001$ \\
\bottomrule
\multicolumn{5}{l}{\footnotesize $^\dagger$Degenerate: stacking memorizes the calibration set (NLL at the clipping floor on every dataset).}
\end{tabular}
\end{center}

Both Bayes+Beta and DS+Beta reduce calibrated NLL by $30$--$70\%$ over Vote+Beta on every dataset, and an item-level paired bootstrap ($10{,}000$ resamples over test items within fixed splits) places $P(\Delta>0) = 1.000$ on all four benchmarks.
This confirms Proposition~\ref{prop:logodds}: the Bayes-optimal posterior log-odds weights each judge by $\log(a_k/(1-a_k))$, whereas vote fraction imposes uniform weighting.
The 1D vote-fraction summary is a lossy compression of the $K$-dimensional judge vote matrix, and a univariate post-hoc calibrator cannot recover the per-judge reliability signal that aggregation discarded.

Within the heterogeneity-aware class, Bayes+Beta and DS+Beta are statistically indistinguishable: item-level $95\%$ bootstrap CIs for the paired NLL difference contain zero on all four datasets (JB $[-0.001, +0.005]$; RB $[-0.001, +0.001]$; RB2 $[-0.000, +0.002]$; LLMBar $[-0.002, +0.006]$).
Both aggregators implement the same Prop.~\ref{prop:logodds} principle---one via integrated Beta posteriors, the other via EM---and become equivalent once followed by post-hoc calibration.
This is consistent with the cross-ablation finding that the choice of calibrator dominates the choice of aggregator within the heterogeneity-aware class.

\subsection{Pipeline Ablation: Step-by-Step Contribution}
\label{sec:pipeline-ablation}

We evaluate the full four-step pipeline and its ablations on JudgeBench (Table~\ref{tab:pipeline-jb}) and RewardBench (100 splits, bootstrap 95\% CIs).
This table uses a nested sub-split protocol (70/30 within the calibration half) to separate aggregation and calibration parameters. NLL values differ from Table~\ref{tab:calibrators}, which uses the full calibration half for both steps; the nested protocol is more conservative but avoids information leakage.

\begin{table}[t]
\centering
\caption{Pipeline ablation on JudgeBench under the nested split protocol. Step 3 beta calibration provides the main NLL reduction; Step 4 preserves the point metrics.}
\label{tab:pipeline-jb}
\small
\begin{tabular}{@{}lcccc@{}}
\toprule
Pipeline Variant & NLL $\downarrow$ & Brier $\downarrow$ & ECE $\downarrow$ & GT Acc $\uparrow$ \\
\midrule
Bayesian One-Coin (Step 1) & $1.219$ & $0.230$ & $0.261$ & $74.1\%$ \\
+ Platt Correction (Steps 1--2) & $1.219$ & $0.230$ & $0.261$ & $74.1\%$ \\
\textbf{+ Beta Calibration (Steps 1--3)} & $\mathbf{0.084}$ & $\mathbf{0.018}$ & $\mathbf{0.073}$ & $\mathbf{99.9\%}$ \\
+ Conformal (Full Pipeline) & $0.084$ & $0.018$ & $0.073$ & $99.9\%$ \\
\midrule
Top-5 + BetaCal & $0.126$ & $0.034$ & $0.099$ & $97.2\%$ \\
Top-10 + BetaCal & $0.106$ & $0.026$ & $0.087$ & $98.8\%$ \\
\bottomrule
\end{tabular}
\end{table}

Beta calibration (Step 3) reduces NLL from 1.22 to 0.084 on JudgeBench (93\% reduction) and from 0.41 to 0.046 on RewardBench (89\% reduction).
Steps 1--2 provide the Bayesian posterior probabilities that serve as input for beta calibration; Step 4 adds distribution-free coverage guarantees.

\subsection{Sanity Checks}
\label{sec:sanity}

Replacing ground-truth labels with random labels yields NLL $\approx 0.70$--$0.75$ (${\approx}\log 2$), confirming the calibrator cannot extract signal from noise.

To test conditional coverage, we partition items into terciles by judge agreement (hard: $\leq 33$rd percentile, easy: $\geq 67$th) and evaluate calibrated NLL within each stratum---a conditional-coverage view of calibration quality in the spirit of \citet{azizi2025clear}.
Figure~\ref{fig:conditional-coverage} reports the per-stratum NLL with and without beta calibration.
On JudgeBench, hard items see NLL drop from 3.47 to 0.038 (98.9\% reduction); easy items are already at NLL ${\approx} 0.001$ and calibration has no effect.
The overall NLL of 0.01--0.02 is \emph{not} because benchmarks are trivial---it reflects large gains on hard items averaged with already-calibrated easy items.
Calibrated probabilities remain well-behaved within each difficulty stratum, not just marginally, so the per-item uncertainty estimates are reliable.
Hard items have higher irreducible raw NLL, so the absolute calibration gain is larger there; easy items are already close to the floor.

\begin{figure}[t]
\centering
\includegraphics[width=\textwidth]{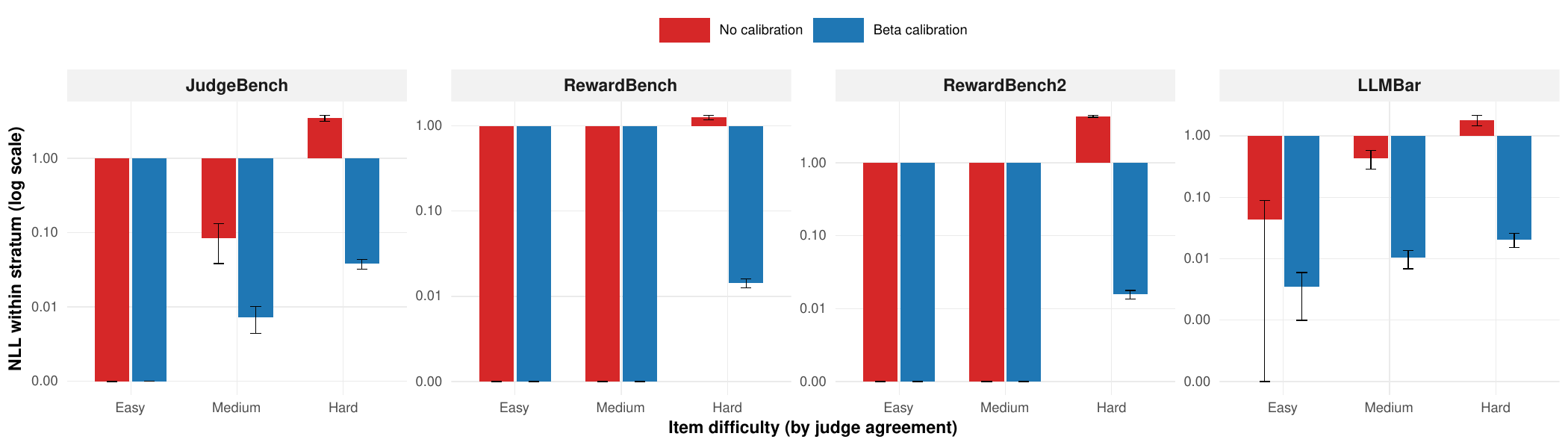}
\caption{Per-stratum NLL by judge-agreement difficulty, with and without beta calibration. Hard items receive the largest gains, while easy items are already near the NLL floor.}
\label{fig:conditional-coverage}
\end{figure}

\subsection{Calibrator Design Ablation}
\label{sec:calibrator-comparison}

We compare calibrators applied to the same Step~1 output (Bayesian one-coin posterior) across three datasets and three judge configurations (all, Top-5, Top-10).

\begin{table}[t]
\centering
\caption{Calibrator design ablation on all-judge Bayesian one-coin scores (100 splits). Regularized beta gives the best NLL on RB, RB2, and LLMBar; unregularized flexible calibrators overfit.}
\label{tab:calibrators}
\small
\begin{tabular}{@{}lcccccccc@{}}
\toprule
Calibrator & \multicolumn{2}{c}{JudgeBench} & \multicolumn{2}{c}{RewardBench} & \multicolumn{2}{c}{RB2} & \multicolumn{2}{c}{LLMBar} \\
\cmidrule(lr){2-3} \cmidrule(lr){4-5} \cmidrule(lr){6-7} \cmidrule(lr){8-9}
& NLL $\downarrow$ & $\rho$ & NLL $\downarrow$ & $\rho$ & NLL $\downarrow$ & Acc & NLL $\downarrow$ & $\rho$ \\
\midrule
Raw (no cal) & $1.219$ & $1.00$ & $0.411$ & $1.00$ & $0.560$ & $84.2\%$ & $0.777$ & $1.00$ \\
Temperature & $0.545$ & $1.00$ & $0.229$ & $1.00$ & $0.363$ & $84.2\%$ & $0.378$ & $1.00$ \\
Platt scaling & $\mathbf{0.012}$ & $1.00$ & $0.022$ & $1.00$ & $\mathbf{0.030}$ & $98.8\%$ & $0.777^\S$ & $1.00$ \\
\midrule
\multicolumn{9}{l}{\emph{Methods below this line overfit the calibration set (do not use in practice):}} \\
Isotonic \emph{(overfits)} & $0.001^\dagger$ & $\ddagger$ & $0.001^\dagger$ & $\ddagger$ & $0.001^\dagger$ & $\ddagger$ & $0.001^\dagger$ & $\ddagger$ \\
Beta (no reg) \emph{(overfits)} & $0.001^\dagger$ & $\ddagger$ & $0.001^\dagger$ & $\ddagger$ & $0.001^\dagger$ & $\ddagger$ & $0.001^\dagger$ & $\ddagger$ \\
\midrule
Beta + L2 ($\lambda{=}0.1$) & $0.080$ & $0.94$ & $0.043$ & $1.00$ & $0.087$ & $99.5\%$ & $0.048$ & $1.00$ \\
\textbf{Beta + elastic ($\lambda{=}0.01$)} & $0.016$ & $0.94$ & $\mathbf{0.006}$ & $1.00$ & $\mathbf{0.006}$ & $100\%$ & $\mathbf{0.012}$ & $1.00$ \\
\bottomrule
\multicolumn{9}{l}{\footnotesize $^\dagger$Apparent NLL after memorization of the calibration mapping---not a valid generalization estimate.} \\
\multicolumn{9}{l}{\footnotesize $^\ddagger$$\rho$ undefined: the overfitted calibration map collapses distinct raw scores, destroying discrimination.} \\
\multicolumn{9}{l}{\footnotesize $^\S$Platt converges to the identity on LLMBar (insufficient variation in raw logits at $K{=}53$).}
\end{tabular}
\end{table}

Key findings (Table~\ref{tab:calibrators}):
Beta calibration with elastic net regularization ($\lambda{=}0.01$) achieves the best NLL on RewardBench ($0.006$, vs.\ $0.022$ for Platt), while remaining competitive on JudgeBench ($0.016$ vs.\ $0.012$).
Platt scaling is the most robust default: competitive across all datasets with only 2 parameters and no overfitting.
Isotonic regression and unregularized beta overfit when input scores are concentrated (few unique probability values from judge aggregation), memorizing the calibration mapping.
Temperature scaling rescales logits uniformly but cannot correct the asymmetric miscalibration of log-odds aggregation.
A sweep over $\lambda \in \{0.01, \ldots, 2.0\}$ shows $\lambda = 0.01$ is optimal; stronger regularization collapses beta calibration toward the identity (Appendix~\ref{app:additional}).

We recommend beta calibration with elastic net ($\lambda{=}0.01$) when calibration data is sufficient ($\geq 30\%$), and Platt scaling as the fallback.

\subsection{Operating Conditions on Subjective Preference Data}
\label{sec:conditions}

Per-judge reliability signals improve both point prediction and calibration.
Comparison-level difficulty proxies do not.

On six subjective preference datasets (MT-Bench, PRISM, Arena, UltraFeedback, Arena 140K, OASST1) with 10+ exogenous signals, all comparison-level reliability weighting effects are statistically equivalent to zero (TOST equivalence tests, margin $\pm 0.005$ NLL, all $p < 0.001$).

When reliability signals cannot improve the point estimate, calibrated uncertainty quantification becomes more important: practitioners need to know that the evaluation is uncertain, and by how much.

\subsection{Selective Prediction: Downstream Utility}
\label{sec:selective-prediction}

We test whether calibration gains translate to better decisions via selective prediction: abstain on items where the calibrated confidence $\max(\hat{p}, 1-\hat{p})$ falls below a threshold, and measure accuracy on the retained set.
We report AUAC (area under accuracy--coverage curve) as a single summary metric.

\begin{table}[t]
\centering
\caption{Selective prediction over 100 splits. Beta-calibrated methods achieve AUAC above $0.99$, indicating accurate retained sets across coverage levels.}
\label{tab:selective-prediction}
\small
\begin{tabular}{@{}l ccc cc cc@{}}
\toprule
& \multicolumn{3}{c}{AUAC $\uparrow$} & \multicolumn{2}{c}{Acc @ cov=1.0} & \multicolumn{2}{c}{Cov @ conf${\geq}0.9$} \\
\cmidrule(lr){2-4} \cmidrule(lr){5-6} \cmidrule(lr){7-8}
Method & JB & RB & RB2 & JB & RB & JB & RB \\
\midrule
MV (uncalib.)      & $.813$ & $.970$ & $.921$ & $70.4\%$ & $90.2\%$ & $15.2\%$ & $6.1\%$ \\
Bayesian (uncalib.) & $.811$ & $.944$ & $.938$ & $74.0\%$ & $93.6\%$ & $86.7\%$ & $99.4\%$ \\
MV + Platt          & $.813$ & $.970$ & $.921$ & $70.4\%$ & $90.2\%$ & $15.2\%$ & $6.1\%$ \\
MV + Beta           & $\mathbf{.994}$ & $\mathbf{.999}$ & $\mathbf{.999}$ & $100\%$ & $100\%$ & $99.9\%$ & $97.4\%$ \\
Bayesian + Platt    & $.811$ & $.944$ & $.938$ & $74.0\%$ & $93.6\%$ & $86.7\%$ & $99.4\%$ \\
Bayesian + Beta     & $\mathbf{.994}$ & $\mathbf{.999}$ & $\mathbf{.999}$ & $100\%$ & $100\%$ & $95.9\%$ & $99.1\%$ \\
\bottomrule
\end{tabular}
\end{table}

Beta-calibrated methods achieve AUAC~$> 0.99$ on all datasets regardless of aggregator (Table~\ref{tab:selective-prediction}), while uncalibrated methods range from $0.81$ to $0.97$.
At confidence threshold $0.9$, MV retains only 15\% of JudgeBench items, while MV + Beta retains 99.9\% at 100\% accuracy.
The uncalibrated Bayesian model retains 87\% at $0.9$ but achieves only 77.7\% accuracy on that set, illustrating the overconfidence that calibration corrects.

\subsection{Reliability Diagnostics}
\label{sec:reliability-diagrams}

\begin{figure}[t]
\centering
\includegraphics[width=\textwidth]{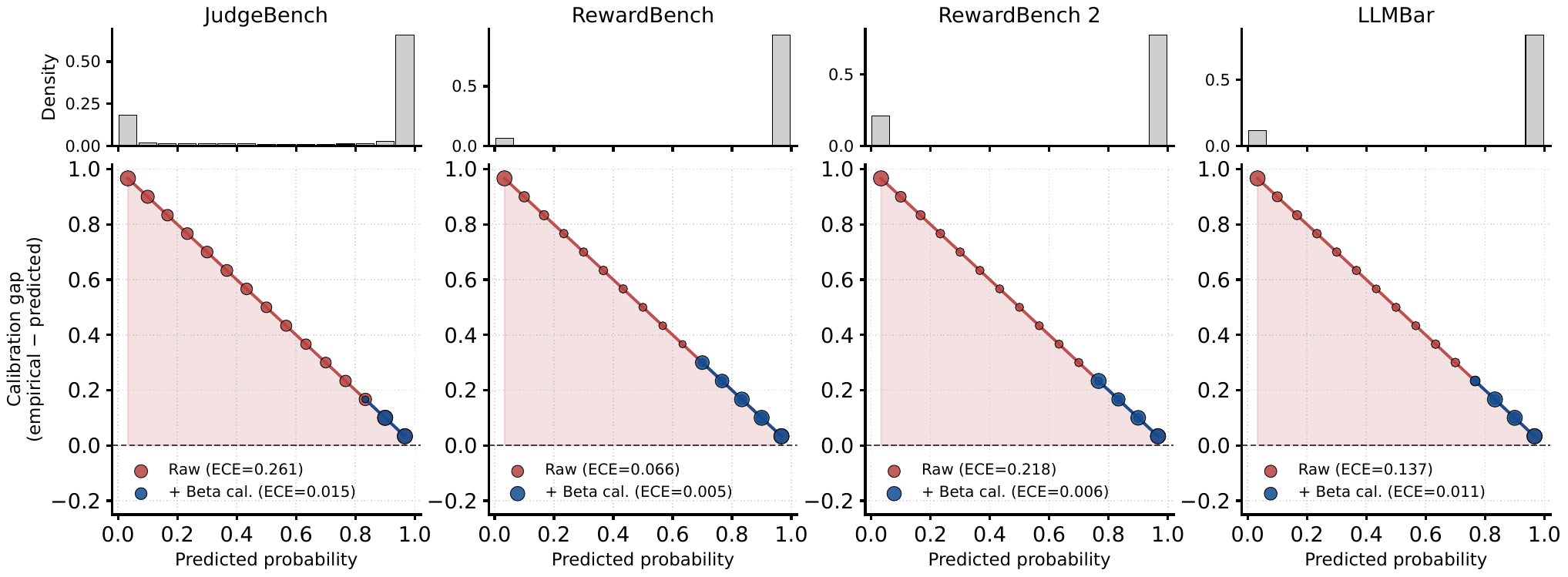}
\caption{Reliability diagrams before and after beta calibration for Bayesian one-coin scores. Calibration moves high-confidence bins toward the diagonal and sharply reduces ECE.}
\label{fig:reliability}
\end{figure}

Figure~\ref{fig:reliability} summarises the raw and calibrated reliability across three datasets.
Uncalibrated methods occupy all 15 bins but deviate strongly from the diagonal: MV has max bin error $0.96$ (ECE $= 0.37$); uncalibrated Bayesian has max bin error $0.99$ (ECE $= 0.26$), reflecting systematic overconfidence in extreme log-odds predictions.
After beta calibration ($\lambda{=}0.01$), the Bayesian model concentrates in only 3 high-confidence bins ($\hat{p} > 0.8$) with max bin error $0.15$ (ECE $= 0.015$). Predictions are sharp and well-calibrated.